\documentclass{article}

\usepackage{amsmath,amssymb,amsbsy,amsgen,amsopn,amsthm}
\usepackage{graphicx}
\usepackage[margin=1in]{geometry}

\title{Approximating the generalized terminal backup problem via
half-integral multiflow relaxation\footnotemark[2]}

\author{Takuro Fukunaga
\thanks{National Institute of Informatics,
2-1-2 Hitotsubashi, Chiyoda-ku, Tokyo, Japan.
JST, ERATO, Kawarabayashi Large Graph Project, Japan.
  {\tt takuro@nii.ac.jp}}}
  \date{}
  
\newtheorem{theorem}{Theorem}
\newtheorem{lemma}{Lemma}

\newtheorem{corollary}{Corollary}

\newtheorem{assumption}{Assumption}

\newcommand{\Cfam}{\mathcal{C}}
\newcommand{\Vfam}{\mathcal{V}}
\newcommand{\Lfam}{\mathcal{L}}

\newcommand{\Afam}{\mathcal{A}}

\newcommand{\Sfam}{\mathcal{S}}
\newcommand{\Zset}{\mathbb{Z}}
\newcommand{\Rset}{\mathbb{Q}}
\newcommand{\Mfam}{\mathcal{M}}
\newcommand{\allone}{\mathbf{1}}
\newcommand{\cut}{P}
\newcommand{\lp}{\mathsf{LP}}
\newcommand{\f}{f^{\kappa}}
\newcommand{\g}{f^{\lambda}}

\begin{document}
\maketitle

\renewcommand{\thefootnote}{\fnsymbol{footnote}}
\footnotetext[2]{An extended abstract of this work appeared in the proceedings of STACS 2015.}

\begin{abstract}
We consider a network design problem called the generalized terminal backup problem.
Whereas earlier work investigated 
the edge-connectivity constraints only,
we consider both edge- and node-connectivity constraints for this problem.
A major contribution of this paper is 
the development of a strongly polynomial-time $4/3$-approximation algorithm 
for the problem. Specifically,
we show that a linear programming relaxation of the problem is half-integral,
and that the half-integral optimal solution can be rounded to a $4/3$-approximate solution.
We also prove that the linear programming relaxation of the problem with the edge-connectivity constraints 
is equivalent to minimizing the cost
of half-integral multiflows that satisfy flow demands given from terminals.
This observation implies a strongly polynomial-time algorithm 
for computing a minimum cost half-integral multiflow
under flow demand constraints.
\end{abstract}

\section{Introduction} 
\label{sec:introduction}

\subsection{Generalized terminal backup problem}
The network design problem is the problem of constructing a low cost network that satisfies
given constraints. It includes many fundamental optimization problems, and 
has been extensively studied. In this paper, we consider a network design problem
called the \emph{generalized terminal backup problem}, recently introduced by Bern{\'a}th and
Kobayashi~\cite{Bernath2014}. 

The generalized terminal backup problem is defined as follows. 
Let $\Rset_+$ and $\Zset_+$ 
denote the sets of non-negative rational numbers and non-negative integers, 
respectively. 
Let $G=(V,E)$ be an undirected graph
with node set $V$ and edge set $E$, 
$c\colon E \rightarrow
\Rset_+$ be an edge cost function,
and let $u\colon E \rightarrow \Zset_+$ be an edge capacity function.
A subset $T$ of $V$ denotes the \emph{terminal} node set in which
each terminal $t$ is associated with a connectivity requirement
$r(t)\in \Zset_+$.
A solution is a multiple edge set on $V$ containing at most $u(e)$ edges parallel 
to $e \in E$. The objective is to find a solution
$F$ that minimizes $\sum_{e \in F} c(e)$ under certain constraints.
In Bern{\'a}th and Kobayashi~\cite{Bernath2014}, 
the subgraph $(V,F)$ was required to contain $r(t)$ 
edge-disjoint paths that connect each $t \in T$ to other terminals.
In addition to these edge-connectivity constraints, we consider
node-connectivity constraints, under which the paths must be 
inner disjoint (i.e., disjoint in edges and nodes in $V\setminus T$) rather than edge-disjoint.
To avoid confusion,
we refer to the problem as \emph{edge-connectivity terminal backup} 
when the edge-connectivity constraints are required, and
as \emph{node-connectivity terminal backup}
when the node-connectivity constraints are imposed.
When $r \equiv 1$, the problem is called the \emph{terminal backup problem}.
Since there is no difference between edge-connectivity and node-connectivity
when $r\equiv 1$, these names make no confusion.

The generalized terminal backup problem 
models a natural data management situation. Suppose that each terminal represents a data storage server
in a network, and $r(t)$ is the amount of data stored in the server at a terminal $t$.
Backup data must be stored in servers different from that storing the original data.
To this end, 
a sub-network that transfers data stored at one terminal to other terminals is required. 
We assume that 
edges can transfer a single unit of data per time unit.
Hence, transferring data from terminal $t$ to other terminals within one time unit
requires $r(t)$ edge-disjoint paths from $t$ 
to $T\setminus \{t\}$, which is represented by the edge-connectivity constraints. 
When nodes are also capacitated, $r(t)$ 
inner-disjoint paths are required; these requirements are met by the node-connectivity
constraints.

The generalized terminal backup problem is interesting also from theoretical point of view.
Anshelevich and Karagiozova~\cite{AnshelevichK11} demonstrated
that the terminal backup problem is reducible to the simplex
matching problem, which is solvable in polynomial time. 
On the other hand, when $T=V$, the generalized terminal backup problem is equivalent to the capacitated 
$b$-edge cover problem with degree lower
bound $b(v)=r(v)$ for $v \in V$. Since the capacitated 
$b$-edge cover problem admits a polynomial-time algorithm, the
generalized terminal backup problem is solvable in polynomial time also when $T=V$. 
Therefore, we may naturally ask whether the generalized terminal backup
problem is solvable in polynomial time. Bern{\'a}th and Kobayashi~\cite{Bernath2014} proposed a
polynomial-time algorithm for the uncapacitated case (i.e., $u(e)=+\infty$ for each $e \in E$)
in the edge-connectivity terminal backup. Their result partially
answers the above question, but their assumptions may be overly stringent in some situations;
that is, their algorithm admits unfavorable solutions that select too many copies of a cheap edge. 
Moreover, their algorithm cannot deal with the node-connectivity constraints.
Unfortunately, 
when the edge-capacities are bounded or node-connectivity constraints
are imposed,
we do not know whether
the generalized terminal backup problem is NP-hard or admits a
polynomial-time algorithm.
Instead, we propose approximation algorithms as follows.

\begin{theorem}\label{thm:main-4/3}
There exist a strongly polynomial-time $4/3$-approximation algorithm for
the generalized terminal backup problem.
\end{theorem}

The present study contributes two major advances
to the generalized terminal backup problem.
\begin{itemize}
	\item Bern{\'a}th and Kobayashi~\cite{Bernath2014} discussed the 
	generalized terminal backup problem 
	in the uncapacitated setting
	with edge-connectivity constraints, noting that 
	the problem in the capacitated setting is open.
	Here, we discuss the capacitated setting, and 
	introduce the node-connectivity constraints.
	\item The generalized
	terminal backup problem can be formulated as the problem of covering skew
	supermodular biset functions, which is known to admit a 2-approximation algorithm.
	On the other hand, as stated in Theorem~\ref{thm:main-4/3}, we develop 
	$4/3$-approximation algorithms, that outperform this 2-approximation algorithm.
\end{itemize}

Let us explain the second advance more specifically.
Given an edge set $F$ and a nonempty subset $X$ of $V$, 
let $\delta_F(X)$ denote the set of edges in $F$ with one end node
in $X$ and the other in $V\setminus X$. 
Let $\g\colon 2^V \rightarrow \Zset_+$ be a function such that 
$\g(X)=r(t)$ if $X\cap T=\{t\}$, and $\g(X)=0$ otherwise. 
By the edge-connectivity version of 
Menger's theorem, $F$ satisfies the edge-connectivity constraints if and only if
$|\delta_F(X)| \geq \g(X)$ for each $X \subset V$. 
Bern{\'a}th and
Kobayashi~\cite{Bernath2014} showed that the function $\g$ is skew supermodular 
(skew supermodularity is defined in Section~\ref{sec:preliminaries}). 
For any skew supermodular set function $h$,
Jain~\cite{Jain01} proposed a seminal $2$-approximation algorithm for computing
a minimum-cost edge set $F$
satisfying $|\delta_F(X)| \geq h(X)$, $X \subset V$.
Although the node-connectivity constraints cannot be captured by set functions 
as the edge-connectivity constraints,
they can be regarded as a request for covering a skew supermodular
\emph{biset} function, to which the 2-approximation algorithm is extended~\cite{FleischerJW06} 
(see Section~\ref{sec:preliminaries}). Therefore, the
generalized terminal backup problem admits 2-approximation algorithms, 
regardless of the imposed connectivity constraints. 
One of our contributions is to improve these 2-approximations to
$4/3$-approximations.

Both of the above 2-approximation algorithms involve 
iterative rounding of the linear programming (LP) relaxations.
Primarily, their performance analyses
prove that the value of a variable in each extreme point solution of the LP relaxations
is at least $1/2$. Once this property of extreme point solutions is proven,
the variables can be repeatedly rounded until 
a 2-approximate solution is obtained. 
Our $4/3$-approximation algorithms are based on the same LP 
relaxations as the iterative rounding algorithms. We
show that, in the generalized terminal backup problem,
all variables in extreme point solutions of the relaxation take 
half-integral values.
We also prove that the half-integral solution can be rounded into an integer solution
with loss of factor at most $4/3$.

It may be helpful for understanding our result to see the well-studied special case of $T=V$ 
and $u(e)=1$ for each $e \in E$ (i.e., feasible solutions are simple $r$-edge covers).
In this case, our LP relaxation minimizes $\sum_{e \in E} c(e)x(e)$ subject to 
$\sum_{e \in \delta(v)}x(e) \geq r(v)$ for each $v \in V$
and $0 \leq x(e) \leq 1$ for each $e \in E$, where $\delta(v)$ is the set of edges incident to the node $v$.
It has been already known that an extreme point solution of this LP 
is half-integral, and the edges in $\{e \in E \colon x(e)=1/2\}$ form odd cycles.
The half-integral variables of the edges on an odd cycle
can be rounded as follows. Suppose that edges $e_1,\ldots,e_k$ 
appear in the cycle in this order, where $k$ is the cycle length (i.e., odd integer larger than one).
For each $i,j \in \{1,\ldots,k\}$,
we define 
$x'_i(e_j)=1$
if $j \geq i$ and $j\equiv i \bmod 2$,
or if $j < i$ and $j\equiv i+1 \bmod 2$,
and 
$x'_i(e_j)=0$ otherwise.
See Figure~\ref{fig.oddcicle}.
for an illustration of this definition.
Note that exactly $(k+1)/2$ variables
in $x'_1(e_j),\ldots,x'_k(e_j)$ are equal to one,
and the other $(k-1)/2$ variables are equal to zero for each $j$.
This means that 
\[
\sum_{i=1}^k \sum_{j=1}^k c(e_j)x'_i(e_j)
= \sum_{j=1}^k c(e_j)\cdot \frac{k+1}{2}
=(k+1)\sum_{j=1}^k c(e_j)x(e_j).
\]
Let $i^*$ minimize $\sum_{j=1}^k c(e_j)x'_{i^*}(e_j)$ in $i^* \in \{1,\ldots,k\}$.
Then, since $\sum_{j=1}^k c(e_j)x'_{i^*}(e_j) \leq \sum_{i=1}^k \sum_{j=1}^k c(e_j)x'_i(e_j)/k$,
replacing $x(e_1),\ldots,x(e_k)$ by $x'_{i^*}(e_1),\ldots,x'_{i^*}(e_k)$
increases their costs by a factor at most $(k+1)/k \leq 4/3$.
We also observe that 
the feasibility of the solution is preserved
even after the replacement.
By applying this rounding for each odd cycle, the half-integral solution 
can be transformed into a $4/3$-approximate integer solution.

\begin{figure}[h]
\centering
\includegraphics[]{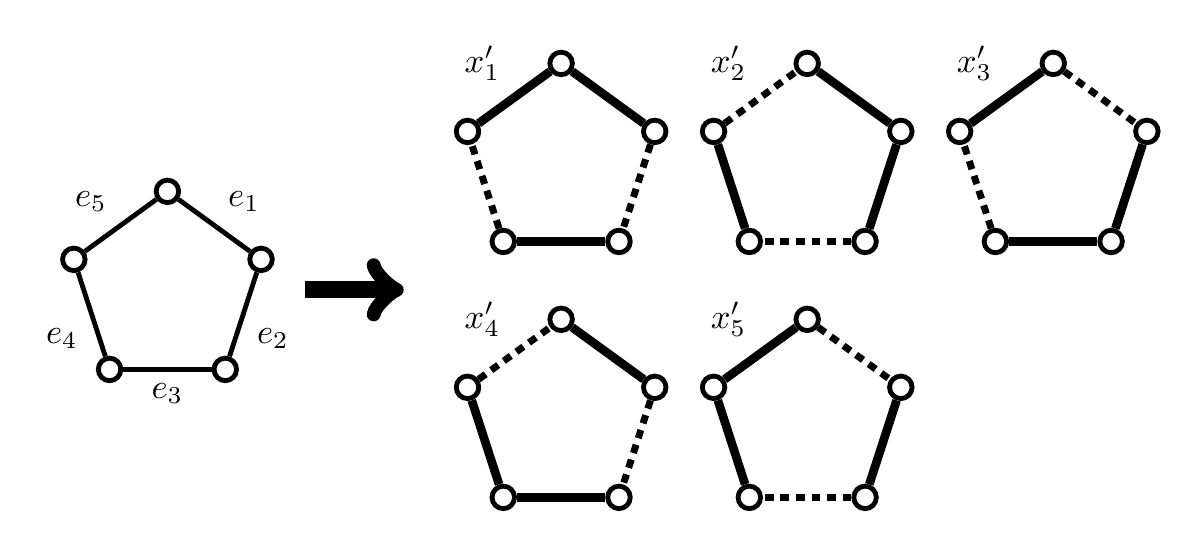}
\caption{Rounding of half-integral variables corresponding to a cycle of length 5. 
A dotted line represents $x'_i(e_j)=0$, and a solid thick line represents $x'_i(e_j)=1$.}
\label{fig.oddcicle}
\end{figure}

Our result is obtained by extending 
the characterization of the edge structure whose corresponding variables are
not integers, but the extension is not immediate.
As in the above special case, those edges form cycles
in the generalized terminal backup problem if the solution is a minimal feasible solution to the LP relaxation.
However, the length of a cycle is not necessarily odd, and it is not clear how the half-integral solution should
be rounded; In the above special case, we round up and down variables of edges on a cycle alternatively, but this 
obviously does not preserve the feasibility in the generalized terminal backup problem.
The key ingredient in our result is to characterize the relationship between the cycles and the node sets or bisets
corresponding to linearly independent tight constraints in the LP relaxation.
We show that a cycle crosses maximal tight node set or bisets an odd number of times,
which extends the property that the length of each cycle is odd in the special case.
Our rounding algorithm decides how to round a non-integer variable 
from the direction of the crossing between the corresponding edge and a tight node set or biset.

\subsection{Minimum cost multiflow problem}\label{subsec.intro-multiflow}

Multiflows are
closely related to the generalized terminal backup problem.
Among the many multiflow variants, we focus on the type
sometimes called \emph{free multiflows}.
For $t,t' \in T$,
$\Afam_{t,t'}$ denotes the set of paths that terminate at $t$ and $t'$.
Let $\Afam_t$ denote $\bigcup_{t' \in T\setminus \{t\}}\Afam_{t,t'}$,
and $\Afam$ denote $\bigcup_{t \in T}\Afam_{t}$.
$E(A)$ and $V(A)$ denote the sets of edges and nodes in $A \in \Afam$, respectively.
We define a multiflow as a function $\psi \colon \Afam \rightarrow \Rset_+$.
In the edge-capacitated setting,
an edge capacity $u(e) \in \Zset_+$ is given,
and we must satisfy 
$\sum\{\psi(A)\colon A \in \Afam, e \in E(A)\} \leq u(e)$ for each $e \in E$.
In the node-capacitated setting, a node capacity $u(v) \in \Zset_+$ is given
and $\sum\{\psi(A)\colon A \in \Afam, v \in V(A)\} \leq u(v)$ is required for each 
$v \in V$.
The multiflow $\psi$ is called an \emph{integral} multiflow
if $\psi(A) \in \Zset_+$ for each $A \in \Afam$,
and is called a \emph{half-integral} multiflow
if $2\psi(A) \in  \Zset_+$ for each $A \in \Afam$.
Let $c(A)$ denote $\sum_{e \in E(A)}c(e)$ for $A \in \Afam$.
The cost of $\psi$ is given by
$\sum_{A \in \Afam}\psi(A)c(A)$.

In the edge-connectivity terminal backup, 
the connectivity requirement from a terminal $t$ equates to requiring that
a flow of amount $r(t)$ can be delivered from $t$ to $T\setminus \{t\}$
in the graph $(V,F)$ with unit edge-capacities if $F$ is a feasible solution.
This condition appears similar to the constraint that 
the graph $(V,F)$ with unit edge-capacities
admits a multiflow $\psi$
such that $\sum_{A \in \Afam_t}\psi(A) \geq r(t)$.
We note that $(V,F)$ with unit edge-capacities admits a multiflow $\psi$
if and only if
the number of copies of $e \in E$ in $F$ is at least $\sum_{A \in \Afam \colon e \in E(A)}\psi(A)$.
These observations suggest a correspondence between the edge-connectivity terminal backup
and the problem of finding a minimum cost multiflow $\psi$
under the constraint that 
$\sum_{A \in \Afam_t}\psi(A) \geq r(t)$ for $t \in T$ in the edge-capacitated setting.
We refer to such a multiflow computation as the \emph{minimum cost multiflow problem}
(in the edge-capacitated setting).
The same correspondence exists between
the node-connectivity terminal backup and the 
node-capacitated setting in the minimum cost multiflow problem.

However,
the generalized terminal backup and the minimum cost multiflow problems are not equivalent.
Especially, the minimum cost multiflow problem can be formulated in LP, whereas 
the generalized terminal backup problem is an integer programming problem.
Even if multiflows are restricted to integral multiflows, 
the two problems are not equivalent.
To observe this, let $G=(V,E)$ be a star with an odd number of leaves. 
We assume that $T$ is the set of leaves,
and each edge incurs one unit of cost.
This star is a feasible solution 
to the terminal backup problem (i.e., $r(t)=1$ for $t \in T$).
In contrast, setting $r\equiv 1$ and $u\equiv 1$
admits no integral multiflow in the edge-capacitated setting, and
no feasible (fractional) multiflows in the node-capacitated setting.

Nevertheless,
similarities exist between terminal backups and multiflows.
As mentioned above, we will show that an LP relaxation of the generalized terminal backup problem
always admits a half-integral optimal solution.
Similarly, half-integrality results are frequently reported for multiflows.
Lov\'asz~\cite{Lovasz76} and Cherkassky~\cite{Cherkasskky77} 
investigated $r \equiv 0$ in the edge-capacitated setting,
and showed that a half-integral multiflow maximizes $\sum_{A \in \Afam}\psi(A)$ over all 
multiflows $\psi$. Using an identical objective function to ours,
Karzanov~\cite{Karzanov94,Karzanov79} sought to
minimize the cost of multiflows.
His feasible multiflow solutions are those attaining $\max \sum_{A \in \Afam}\psi(A)$
in the edge-capacitated setting with $r \equiv 0$,
and he showed that the minimum cost is achieved by a half-integral multiflow.
Babenko and Karzanov~\cite{BabenkoK12} and Hirai~\cite{Hirai13} extended 
Karzanov's result to node-cost minimization in the node-capacitated setting.
In this scenario also, the optimal multiflow is half-integral.

In the present paper, we present a useful relationship 
between the generalized terminal backup problem
and
the minimum cost multiflow problem in the edge-capacitated setting.
% In this paper, we show that 
% the edge-capacitated minimum cost multiflow problem admits a half-integral optimal solution, and
% it can be computed in strongly polynomial time. 
We prove that the optimal solution of the LP
used to approximate the edge-connectivity terminal backup is a half-integral multiflow,
which also optimizes the minimum cost multiflow problem.
Thereby, we can compute the minimum cost half-integral multiflow by
solving the LP relaxation. This result is summarized in the following theorem.

\begin{theorem} \label{thm:flow}
The minimum cost multiflow problem admits a half-integral optimal solution in the edge-capacitated 
setting, which can be computed in strongly polynomial time.
\end{theorem}

In contrast, we find no useful relationship between the node-connectivity terminal backup
and the node-capacitated setting of the minimum cost multiflow problem.
We can only show that the LP relaxation of the node-connectivity terminal backup also
has an optimal solution which is a half-integral multiflow in the edge-capacitated setting.

Despite its natural formulation,
the minimum cost multiflow problem has not been previously investigated to our knowledge.
We emphasize that Theorem~\ref{thm:flow} cannot be derived from previously known results on multiflows.
The minimum cost multiflow problem
may be solvable 
by reducing it to minimum cost maximum multiflow problems that (as mentioned above)
admit polynomial-time algorithms.
A naive reduction can be implemented as follows.
Let $\psi^*$ be a minimum cost multiflow that satisfies the flow demands from terminals, and let
$\nu(t)=\sum_{A \in \Afam_t}\psi^*(A)$
for each $t \in T$. 
For each $t \in T$,
we add a new node $t'$ and connect $t$ and $t'$ by a new edge of capacity $\nu(t)$.
The new terminal set $T'$ is defined as $\{t' \colon t \in T\}$.
Now the multiflow $\psi^*$ can be extended to the multiflow of maximum flow value for the terminal
set $T'$. Applying
the algorithm in~\cite{Karzanov94} to this new instance,
we can solve the original problem.
Moreover, if $\nu(t)$ is an integer for each $t \in T$,
this reduction together with the half-integrality result in \cite{Karzanov79,Karzanov94} implies that 
an optimal multiflow in the minimum cost multiflow problem is half-integral.
However, this naive reduction has two limitations. First, $\nu(t)$ is indeterminable
without computing $\psi^*$. We only know that $\nu(t)$ cannot be smaller than $r(t)$.
Second, we cannot ascertain that $\nu(t)$ is always an integer for each $t \in T$.
Hence, this naive reduction
seems to yield neither a polynomial-time algorithm nor the half-integrality of optimal
multiflows claimed in Theorem~\ref{thm:flow}.

Applying a structural result in~\cite{Bernath2014} on the generalized terminal backup problem, 
it is easily shown that any integral solution to the edge-connectivity terminal backup
provides a half-integral multiflow at the same cost. 
However, since the way to find an optimal solution for the edge-connectivity terminal backup is
unknown, Theorem~\ref{thm:flow} is not derivable from this relationship.
In proving the half-integrality of the LP
relaxation required for Theorem~\ref{thm:main-4/3},
we immediately imply the quarter-integrality of a minimum cost multiflow (i.e.,
$4\psi(A)\in \Zset_+$ for each 
$A \in \Afam$). The proof of Theorem~\ref{thm:flow}
requires deeper investigation into the structure of half-integral LP solutions.

\subsection{Structure of this paper}

Section~\ref{sec:preliminaries} introduces notations and essential preliminaries on bisets.
Section~\ref{sec:characterization} proves that an LP relaxation of the generalized terminal backup
problem admits half-integral optimal solutions, and characterizes the edges assigned with half-integral
values.
Section~\ref{sec.algorithm} introduces our $4/3$-approximation algorithm for the generalized terminal
backup problem, which proves Theorem~\ref{thm:main-4/3}.
Section~\ref{sec:half_integralty_of_minimum_cost_multiflow}
discusses relationship between the generalized terminal backup 
and the minimum cost multiflow problems with a proof of Theorem~\ref{thm:flow}.
Section~\ref{sec.conclusion} concludes the paper.

\section{Preliminaries} 
\label{sec:preliminaries}

\subsection{Bisets}
A biset $\hat{X}$ is defined as an ordered pair $(X,X^+)$ of node sets $X$ and $X^+$ with 
$X \subseteq X^+ \subseteq V$. 
The former and latter elements are respectively called the \emph{inner part} and \emph{outer
part} of the biset. 
Throughout the paper, we denote the inner part of a biset $\hat{X}$ by $X$, and the outer
part by $X^+$. 
$X^+\setminus X$ is called the \emph{neighbor} of
$\hat{X}$, and is denoted by $\Gamma(\hat{X})$. 
$\Vfam$ is the family of all bisets with nonempty inner parts of $V$. 
For an edge set $F$ and a biset $\hat{X}$, $\delta_F(\hat{X})$
denotes the set of edges in $F$ with one end node in $X$ and the other in $V \setminus X^+$. 
We identify a node $v \in V$ with the biset $(\{v\},\{v\})$. Thereby $\delta_F(v)$ denotes the set
of edges incident to $v$ in $F$.
For simplicity, we
write $\delta_E(\hat{X})$ as $\delta(\hat{X})$ when the edge set is unambiguously $E$.
If an edge $e$ is in $\delta(\hat{X})$, we say that $e$ is \emph{incident} to $\hat{X}$.

For two bisets $\hat{X}$ and $\hat{Y}$, we define $\hat{X} \cap \hat{Y}$ as $(X\cap Y, X^+ \cap
Y^+)$, $\hat{X} \cup \hat{Y}$ as $(X \cup Y, X^+ \cup Y^+)$, and $\hat{X} \setminus \hat{Y}$ as 
$(X \setminus Y^+, X^+ \setminus Y)$. 
If $X \subseteq Y$ and $X^+ \subseteq Y^+$, then we write $\hat{X} \subseteq \hat{Y}$. 
This inclusion relationship defines a partial order on the bisets, from which
we define the maximality and minimality among the bisets.

We say that $\hat{X}$ and $\hat{Y}$ are 
\emph{strongly disjoint} when $X \cap Y^+ = \emptyset = X^+ \cap Y$.
 If $\hat{X}$ and $\hat{Y}$ are strongly disjoint, $\hat{X} \setminus \hat{Y} = \hat{X}$ and 
 $\hat{Y}\setminus \hat{X} = \hat{Y}$.
$\hat{X}$ and $\hat{Y}$ are called \emph{noncrossing} when strongly disjoint,
$\hat{X}\subseteq \hat{Y}$, 
or when $\hat{Y} \subseteq \hat{X}$.
Otherwise, $\hat{X}$ and $\hat{Y}$ are called \emph{crossing}.
A family of bisets is called \emph{laminar} if each pair of bisets in the family is noncrossing.
The laminarity naturally defines a child-parent relationship among bisets (or a forest structure on
bisets). Let $\Lfam$ be a laminar family of bisets in $\Vfam$. If $\hat{X}, \hat{Y}, \hat{Z} \in
\Lfam$ satisfy $\hat{X}\subseteq \hat{Y}$ and $\hat{X} \subseteq \hat{Z}$, laminarity implies
that $\hat{Y}\subseteq \hat{Z}$ or $\hat{Z}\subseteq\hat{Y}$. 
Hence, each $\hat{X} \in \Lfam$ admits a unique minimal biset $\hat{Y} \in \Lfam$ with $\hat{X}
\subseteq \hat{Y}$
unless $\hat{X}$ is maximal in $\Lfam$. Such a biset $\hat{Y}$ is defined as the \emph{parent} of 
$\hat{X}$, and $\hat{X}$ is a \emph{child} of 
$\hat{Y}$. This child-parent relationship
naturally leads to terminologies such as ``ancestor'' and ``descendant.'' 
For a biset $\hat{Y}$ in a laminar family $\Lfam$ and an edge set $F$, 
we let $F^+_{\Lfam}(\hat{Y})$ 
and $F^-_{\Lfam}(\hat{Y})$ 
 respectively denote 
$\delta_F(\hat{Y}) \setminus (\bigcup_{\hat{X} \in \mathcal{X}}\delta_F(\hat{X}))$
and 
$(\bigcup_{\hat{X} \in \mathcal{X}}\delta_F(\hat{X}))\setminus \delta_F(\hat{Y})$, where $\mathcal{X}$
denotes the set of children of $\hat{Y}$ in $\Lfam$.
If $\hat{Y}$ has no child, 
$F^+_{\Lfam}(\hat{Y})=\delta_F(\hat{Y})$ and
$F^-_{\Lfam}(\hat{Y})=\emptyset$.

\subsection{Bisets and connectivity of graphs}

For $t \in T$, let 
\[
\Cfam(t)=\{\hat{X}\in \Vfam \colon X\cap T=X^+\cap T=\{t\}\}.
\]
We denote $\bigcup_{t \in T}\Cfam(t)$ by $\Cfam$.
For a vector $x \in \Rset_+^E$ and $E' \subseteq E$, 
let $x(E')$ represent $\sum_{e \in E'}x(e)$.
We define a biset function $\f$ by
 \[
   \f(\hat{X}) = 
   \begin{cases}
     r(t)-|\Gamma(\hat{X})|, &\text{ if } \hat{X} \in \Cfam(t) \text{ for some } t \in T,\\
     0, &\text{ otherwise }
   \end{cases}
 \]
for each $\hat{X} \in \Vfam$.
According to the node-connectivity version of Menger's theorem, 
the graph $(V,F)$ contains $r(t)$ inner-disjoint paths between $t$ and $T\setminus
\{t\}$
if and only if
$|\delta_{F}(\hat{X})| + |\Gamma(\hat{X})| \geq r(t)$ 
for each $\hat{X} \in \Cfam(t)$.
This condition is equivalent to $|\delta_{F}(\hat{X})| \geq \f(\hat{X})$
for all $\hat{X} \in \Vfam$.

In Section~\ref{sec:introduction},
we defined the set function $\g$ representing 
the edge-connectivity constraints.
For treating both node-connectivity and edge-connectivity simultaneously,
we sometimes extend
$\g$ to a biset function by
identifying $X\subseteq V$ with the biset $(X,X)$.
Specifically, the biset function $\g$ is defined by
\[
  \g(\hat{X}) = 
  \begin{cases}
    r(t), &\text{ if }  t \in T, \hat{X} \in \Cfam(t),
    \Gamma(\hat{X})=\emptyset,\\
    0, &\text{ otherwise }
  \end{cases}
\]
for each $\hat{X} \in \Vfam$.

Given a biset function $h$ and 
an edge-capacity function
$u\colon E \rightarrow \Zset_+$,
we define $\cut(h,u)$ as the set of 
$x \in \Rset_+^E$ such that 
\begin{equation}\label{eq:t-cut}
x(\delta(\hat{X})) \geq h(\hat{X})  \ \ \ \mbox{for $\hat{X}\in \Vfam$}
\end{equation}
and 
\[
x(e) \leq u(e) \text{ for } e \in E.
\]

Let $F$ be a multiset of edges in $E$, and $\chi_F$ denote the characteristic vector of $F$
(i.e., $\chi_F \in \Zset_+^E$ and $F$ contains $\chi_F(e)$ copies of $e$ for each $e \in E$).
Note that $|\delta_F(\hat{X})| = \chi_F(\delta(\hat{X}))$ for $\hat{X} \in \Vfam$.
Hence, $\chi_F \in \cut(\f,u)$
if and only if $F$ is a feasible solution
to the node-connectivity terminal backup.
Similarly, $\chi_F \in \cut(\g,u)$
if and only if $F$ is a feasible solution 
to the edge-connectivity terminal backup.
These statements imply that the LP 
 $\lp(h,u)=\min \left\{ \sum_{e\in E} c(e)x(e) \colon x \in \cut(h,u) \right\}$
relaxes the node-connectivity and the edge-connectivity terminal backups 
when $h=\f$ and $h=\g$, respectively.

A biset function $h$ is called \emph{{\rm (}positively{\rm )} skew supermodular} when, 
for any $\hat{X} \in \Vfam$ with $h(\hat{X})>0$ and $\hat{Y}\in \Vfam$ with $h(\hat{Y})>0$, 
$h$ satisfies 
\begin{equation}\label{eq:supermodular}
h(\hat{X})+h(\hat{Y}) \leq h(\hat{X}\cap \hat{Y}) + h(\hat{X} \cup \hat{Y})
\end{equation}
or
\begin{equation}\label{eq:negamodular}
h(\hat{X})+h(\hat{Y}) \leq h(\hat{X}\setminus \hat{Y}) + h(\hat{Y} \setminus \hat{X}).
\end{equation}
For any biset function $h$ and a vector
$x\colon E \rightarrow \Rset_+$,
we let $h_{x}$ denote the biset function such that
$h_{x}(\hat{X})=h(\hat{X})-x(\delta(\hat{X}))$ for each $\hat{X} \in \Vfam$.
The skew supermodularity of $\g_{x}$ was reported by Bern{\'a}th and
Kobayashi~\cite{Bernath2014}. Here, we prove that $\f_{x}$ is also skew supermodular.

\begin{theorem}\label{thm:skewsupermodular}
The biset function $\f_{x}$ is skew supermodular for any $x\colon E \rightarrow \Rset_+$.
\end{theorem}
\begin{proof}
Let $\hat{X}$ and $\hat{Y}$ be two bisets.
$\hat{X}$ and $\hat{Y}$ are known to always satisfy
$|\Gamma(\hat{X})|+|\Gamma(\hat{Y})| \geq |\Gamma(\hat{X}\cap \hat{Y})|+|\Gamma(\hat{X}\cup \hat{Y})|$,
$|\Gamma(\hat{X})|+|\Gamma(\hat{Y})| \geq |\Gamma(\hat{X}\setminus \hat{Y})|+|\Gamma(\hat{Y}\setminus
\hat{X})|$,
$x(\delta(\hat{X}))+x(\delta(\hat{Y})) \geq x(\delta(\hat{X}\cap \hat{Y}))+x(\delta(\hat{X}\cup
\hat{Y}))$,
and
$x(\delta(\hat{X}))+x(\delta(\hat{Y})) 
\geq x(\delta(\hat{X}\setminus \hat{Y})) + x(\delta(\hat{Y}\setminus \hat{X}))$.
These inequalities can be proven by counting contributions of edges on both sides.

Suppose that $f_{x}(\hat{X}) >0$ and $f_{x}(\hat{Y})>0$.
Then $\hat{X},\hat{Y} \in \Cfam$.
If $\hat{X}, \hat{Y} \in \Cfam(t)$ for some $t \in T$, then both $\hat{X}\cap \hat{Y}$ and 
$\hat{X}\cup \hat{Y}$ belong to $\Cfam(t)$.
From this statement and the above inequalities, we have 
$f_{x}(\hat{X})+f_{x}(\hat{Y})
\leq f_{x}(\hat{X} \cap \hat{Y})+f_{x}(\hat{X}\cup \hat{Y})$ in this case.
If $\hat{X} \in \Cfam(t)$ and $\hat{Y} \in \Cfam(t')$ for some $t,t' \in T$ with $t\neq t'$, 
then $\hat{X}\setminus \hat{Y} \in \Cfam(t)$ and $\hat{Y}\setminus \hat{X} \in \Cfam(t')$.
In this case, we have 
$f_{x}(\hat{X})+f_{x}(\hat{Y})
\leq f_{x}(\hat{X} \setminus \hat{Y})+f_{x}(\hat{Y}\setminus \hat{X})$.
\end{proof}

\section{Structure of extreme point solutions}
\label{sec:characterization}

In this section, we present the properties of the extreme points of $\cut(\f,u)$ and $\cut(\g,u)$.
More precisely, we prove that each extreme point of $\cut(\f,u)$ and $\cut(\g,u)$ is half-integral,
and that the edges whose corresponding variables are not integers are characteristically structured.
Note that both $\f$ and $\g$ are integer-valued skew supermodular functions, and 
$\f(\hat{X})=\g(\hat{X}) = 0$
for any $\hat{X} \not\in \Cfam$. 
In the following, we denote an integer-valued skew supermodular function by $h$,
and an extreme point of $\cut(h,u)$ by $x$.

\subsection{Half-integrality} 

Given an edge set $F$ on $V$ and $\hat{X} \in \Vfam$, let $\eta_{F,\hat{X}}$ denote the characteristic
vector of $\delta_F(\hat{X})$,
i.e., an $|F|$-dimensional vector whose components are set to $1$ if indexed by
an edge in $\delta_F(\hat{X})$, and $0$ otherwise.
The following lemma has been previously proposed~\cite{Cheriyan2006,FleischerJW06}.

\begin{lemma}\label{lem.terminal-laminar}
Let $h$ be a skew supermodular biset function, 
and $x$ be an extreme point of $\cut(h,u)$.
Let $E_0=\{e \in E \colon x(e)=0\}$,
$E_1=\{e \in E \colon x(e)=u(e)\}$, and
$F=E \setminus (E_0 \cup E_1)$. 
Let $\Lfam$ be an inclusion-wise maximal laminar subfamily
of $\{\hat{X} \in \Vfam \colon x(\delta_F(\hat{X}))=h(\hat{X})-u(\delta_{E_1}(\hat{X}))>0\}$
such that the vectors in $\{\eta_{F,\hat{X}}\colon \hat{X} \in \Lfam\}$
are linearly independent. Then
$|F|=|\Lfam|$,
and $x$ is a unique vector that satisfies
$x(\delta_F(\hat{X}))=h(\hat{X})-u(\delta_{E_1}(\hat{X}))>0$ for each $\hat{X} \in \Lfam$,
 $x(e)=0$ for each $e \in E_0$, and $x(e)=u(e)$ for each $e \in E_1$.
 Moreover, if some $\hat{Y} \not\in \Lfam$  satisfies
 $x(\delta_F(\hat{Y}))=h(\hat{Y})-u(\delta_{E_1}(\hat{Y}))>0$,
 then $\eta_{F,\hat{Y}}$ is represented as a convex combination of vectors 
 $\eta_{F,\hat{X}}$,  $\hat{X} \in \Lfam$.
\end{lemma}

We note that $\Lfam$ in Lemma~\ref{lem.terminal-laminar} can be
constructed from the extreme point solution $x$ in a greedy
way;
initialize $\Lfam$ to an empty set, and repeatedly add a biset $\hat{X}$
such that $x(\delta_F(\hat{X}))=h(\hat{X})-u(\delta_{E_1}(\hat{X}))>0$, $\eta_{F,\hat{X}}$ is linearly independent from the 
vectors defined from the bisets in the current $\Lfam$, and adding $\hat{X}$ to $\Lfam$ preserves
laminarity of $\Lfam$.
Hereafter, we assume that $\Lfam$ is constructed as claimed in Lemma~\ref{lem.terminal-laminar}.
Similarly, 
$E_0$, $E_1$, and $F$ are defined from $x$ as in Lemma~\ref{lem.terminal-laminar}.

Let $\bar{x}\colon E \rightarrow \Zset_+$,
and define a biset function 
$h_{\bar{x}}(\hat{X})=h(\hat{X})-\bar{x}(\delta(\hat{X}))$ for $\hat{X} \in \Vfam$.
Let $\allone$ denote the $|E|$-dimensional all-one vector.
The following lemma relates only to the extreme points of $\cut(h_{\bar{x}},\allone)$.
In Corollary~\ref{cor.half-integrality}, we will show
that this is sufficient for proving the half-integrality of $\cut(h,u)$.
If $h(\hat{X}) > 0$ holds only for $\hat{X} \in \Cfam$, 
we have $\Lfam \subseteq \Cfam$. In this case, no biset in $\Lfam$ 
has more than one child, and $x$ is characterized as follows.

\begin{lemma}\label{lem.terminal-characterization}
Suppose that $h$ is an integer-valued skew supermodular biset function 
such that $h(\hat{X})>0$ only for $\hat{X} \in \Cfam$.
Let $\bar{x}\colon E \rightarrow \Zset_+$,
and let $x$ be an extreme point of $\cut(h_{\bar{x}},\allone)$.
Let $F$ denote $\{e \in E\colon 0 < x(e) < 1\}$.
 Then the following conditions hold\/{\rm :}
 \begin{enumerate}
  	\item[\rm (i)] $|F^+_{\Lfam}(\hat{X})|+ |F^-_{\Lfam}(\hat{X})| =2$ 
	     for each $\hat{X} \in \Lfam$\/{\rm ;}
  	\item[\rm (ii)] If $e \in F$ is incident to a maximal biset in $\Lfam$, then it is incident to
  	exactly two maximal bisets in $\Lfam$\/{\rm ;}
  	\item[\rm (iii)] $x(e)=1/2$ for each $e \in F$.
 \end{enumerate}
\end{lemma}

\begin{proof}
We first prove (i) and (ii) by contradiction. Let us assume that not all of these
conditions hold. For each pair of $e \in F$ and its end node $v$, we distribute a token
to a biset in $\Lfam$. The biset that obtains the token corresponding to $(e,v)$ is decided
as follows:
\begin{itemize}
 \item If there exist one or more bisets $\hat{X}\in \Lfam$ such that 
	$e \in \delta_F(\hat{X})$ and $v \in X$, the token is assigned to the minimal of these bisets.
\item Otherwise, the token is assigned to the minimal biset $\hat{Y}$ that 
includes both end nodes of $e$ in its outer part (if such a biset exists).
	Notice that such a minimal biset is unique because $\Lfam$ is laminar and $e$ is incident to at
	least one biset in $\Lfam$.
\end{itemize}

The total number of tokens is at most $2|F|$. In the following, 
we prove that tokens may be rearranged so that each
biset in $\Lfam$ receives at least two tokens and at least one biset receives three tokens. 
This rearrangement implies that the number of tokens exceeds
$2|\Lfam|$, contradicting our requirement that $|\Lfam|=|F|$.

Recall that $E_1=\{e \in E \colon x(e)=1\}$.
Let $\bar{x}'$ denote $\bar{x}+\chi_{E_1}$, and
let $\hat{X}$ be a minimal biset in $\Lfam$. 
The minimality of $\hat{X}$ implies $F^-_{\Lfam}(\hat{X})= \emptyset$
and $F^+_{\Lfam}(\hat{X})= \delta_F(\hat{X})$.
Since $x(\delta_F(\hat{X}))=h_{\bar{x}'}(\hat{X}) >0$ and $x(e) < 1$ for each 
$e \in \delta_F(\hat{X})$, we have 
$|F^+_{\Lfam}(\hat{X})|=|\delta_F(\hat{X})| \geq 2$. 
Since each edge in $\delta_F(\hat{X})$ allocates one token to $\hat{X}$, $\hat{X}$
obtains at least two tokens. If $\hat{X}$ violates (i), 
then $|F^+_{\Lfam}(\hat{X})|=|\delta_F(\hat{X})| \geq 3$, and
$\hat{X}$ obtains at least three tokens.

Next, let $\hat{X}$ be a biset in $\Lfam$ that admits a child $\hat{Y} \in \Lfam$.
Since $\eta_{F,\hat{X}}$ and $\eta_{F,\hat{Y}}$ are linearly independent,
$|F_{\Lfam}^+(\hat{X})| + |F^-_{\Lfam}(\hat{X})|> 0$.
Therefore,
if $h_{\bar{x}'}(\hat{X})=h_{\bar{x}'}(\hat{Y})$, then
$|F_{\Lfam}^+(\hat{X})|\geq 1$ and $|F^-_{\Lfam}(\hat{X})| \geq 1$.
If $h_{\bar{x}'}(\hat{X}) > h_{\bar{x}'}(\hat{Y})$, then 
$|F_{\Lfam}^+(\hat{X})| \geq 2$ because $x(e)<1$, 
$e \in F_{\Lfam}^+(\hat{X})$. Similarly,
if $h_{\bar{x}'}(\hat{X}) < h_{\bar{x}'}(\hat{Y})$, then 
$|F_{\Lfam}^- (\hat{X})| \geq 2$.
In summary, either case yields
$|F_{\Lfam}^+(\hat{X})| + |F^-_{\Lfam}(\hat{X})|\geq 2$.
Since $\hat{X}$ receives a token from
each edge in $F_{\Lfam}^+(\hat{X}) \cup F^-_{\Lfam}(\hat{X})$, it obtains
at least two tokens and 
at least three tokens if condition (i) is violated.

Extending the above discussion, each biset in $\Lfam$ obtains at least two tokens, implying that the number of tokens is at least $2|\Lfam|$. If (i) is violated
for any biset in $\Lfam$,
that biset receives more than two tokens. Now suppose that (ii) is violated. Then 
there exists an edge $e \in F$ incident to exactly one maximal biset $\hat{X}$ in $\Lfam$. 
The relation $e \in \delta_F(\hat{X})$ indicates that $e$ has an end node $v \in V \setminus X^+$,
and the token corresponding to $(e,v)$ is assigned to no biset in $\Lfam$.
Therefore, if either (i) or (ii) is violated, the number of tokens exceeds the required $2|\Lfam|$.

Let $y \in \Rset_+^E$ be the vector with components $y(e)=1/2$ for each $e \in F$, and $y(e)=x(e)$
for each $e \in E\setminus F$.
Let $\hat{X} \in \Lfam$, and denote the child of $\hat{X}$ (if it exists) by $\hat{Y}$.
From the above discussion, we obtain the following statements:
\begin{itemize}
\item $h_{\bar{x}}(\hat{X})=1$ and $|\delta_F(\hat{X})|=2$ if $\hat{X}$ is minimal;
\item $|F^+_{\Lfam}(\hat{X})|=|F^-_{\Lfam}(\hat{X})|=1$ if
$\hat{X}$ is not minimal and $h_{\bar{x}}(\hat{X})=h_{\bar{x}}(\hat{Y})$;
\item 
$|F^+_{\Lfam}(\hat{X})|=2$, $|F^-_{\Lfam}(\hat{X})|=0$ 
and $h_{\bar{x}}(\hat{X})=h_{\bar{x}}(\hat{Y})+1$
if
$\hat{X}$ is not minimal and $h_{\bar{x}}(\hat{X})>h_{\bar{x}}(\hat{Y})$;
\item 
$|F^+_{\Lfam}(\hat{X})|=0$, $|F^-_{\Lfam}(\hat{X})|=2$, 
and $h_{\bar{x}}(\hat{X})+1=h_{\bar{x}}(\hat{Y})$
if
$\hat{X}$ is not minimal and $h_{\bar{x}}(\hat{X}) < h_{\bar{x}}(\hat{Y})$.
\end{itemize}
Therefore, $y$ satisfies $y(\delta(\hat{X}))=h_{\bar{x}}(\hat{X})$ 
for each $\hat{X} \in \Lfam$. 
Since this condition is also uniquely satisfied by vector $x$,
we have $x=y$, which proves (iii).
\end{proof}

\begin{corollary}\label{cor.half-integrality}
Suppose that $h$ is a skew supermodular biset function such that $h(\hat{X})>0$ 
only if $\hat{X} \in \Cfam$. Let $u \colon E \rightarrow \Zset_+$.
Given $x \in \cut(h,u)$, we define 
$\bar{x}\colon E \rightarrow \Zset_+$ and $x' \colon E \rightarrow \Rset_+$ 
by $\bar{x}(e)=\lfloor x(e) \rfloor$ 
and $x'(e)=x(e)-\bar{x}(e)$, respectively for each $e \in E$.
If $x$ is an extreme point of $\cut(h,u)$,
then $x'$ is an extreme point of $\cut(h_{\bar{x}},\allone)$.
Moreover, $\cut(h,u)$ is half-integral if $h$ is integer-valued.
\end{corollary}
\begin{proof}
Note that $0\leq x'(e) < 1$ for $e \in E$
and $x'(\delta(\hat{X}))
=x(\delta(\hat{X}))-\bar{x}(\delta(\hat{X}))
\geq h(\hat{X})-\bar{x}(\delta(\hat{X}))=h_{\bar{x}}(\hat{X})$
for $\hat{X} \in \Vfam$.
Hence, $x' \in \cut(h_{\bar{x}},\allone)$.
In the following, we show that $x'$ is an extreme point of $\cut(h_{\bar{x}},\allone)$
if $x$ is an extreme point of $\cut(h,u)$.
This proves that $x$ is half-integral because
$\cut(h_{\bar{x}},\allone)$ is half-integral by Lemma~\ref{lem.terminal-characterization}.

If $x'$ is not an extreme point of 
$\cut(h_{\bar{x}},\allone)$, there exist $y,y' \in \cut(h_{\bar{x}},\allone)$ 
and a real number $\alpha$
such that $x'=\alpha y + (1- \alpha)y'$ and $0 < \alpha < 1$.
Then, $x=x'+\bar{x}=\alpha (y+\bar{x}) + (1- \alpha)(y'+\bar{x})$.
Note that both of $y+\bar{x}$ and $y'+\bar{x}$ are contained in $\cut(h,u)$,
implying that $x$ is not an extreme point of $\cut(h,u)$.
\end{proof}

\subsection{Path decompositions of extreme point solutions} 
\label{subsec.path-decomp}

We denote $\{\hat{X} \in \Lfam \colon t \in X\}$ by $\Lfam(t)$ for each $t \in T$.
Let $t \in T$ with $\Lfam(t) \neq \emptyset$,
and let $\hat{X}_t$ be the maximal biset in $\Lfam(t)$. 
We obtain a graph $G^s[\hat{X}_t]$ 
from $G$ by shrinking all the nodes in $V\setminus X^+_t$ into a single node $s$.
Removing $s$ from $G^s[\hat{X}_t]$,
we obtain another graph $G[\hat{X}_t]$
(i.e., $G[\hat{X}_t]$ is the subgraph of $G$ induced by $X^+_t$).
We suppose that each edge $e$ in 
$G^s[\hat{X}_t]$ or in $G[\hat{X}_t]$ is capacitated by $x(e)$.
If $h=\f$, 
each node $v$ in $G^s[\hat{X}_t]$ except $s$ and $t$ has
unit capacity. When $h=\g$,
each node has unbounded capacity.
The capacities of $s$ and $t$ are always unbounded.
Since all capacities are half-integral, the maximum flow 
between $s$ and $t$ in $G^s[\hat{X}_t]$
can be decomposed into a set of paths 
$R^t_1,\ldots,R^t_{2r(t)}$ each of which accommodates a half unit of flow. 

Let $\hat{Y} \in \Lfam(t)$.
Each path between $s$ and $t$ 
passes through an edge in $\delta(\hat{Y})$ or a node in $\Gamma(\hat{Y})$.
Since $x(\delta(\hat{Y}))+|\Gamma(\hat{Y})|=r(t)$, the edges in $\delta(\hat{Y})$ 
and nodes in $\Gamma(\hat{Y})$ are used to full capacity by the maximum flow,
and each path $R^t_i$ includes exactly one edge in $\delta(\hat{Y})$
or one node in $\Gamma(\hat{Y})$.

Suppose that both $R^t_i$ and $R^t_j$ include a node $v \not\in \{s,t\}$.
Let $e_i$ and $e'_i$ be the edges incident to $v$ on $R^t_i$, where
$e_i$ is near to $s$ than $e'_i$.
We define the edges $e_j$ and $e'_j$ incident to $v$ on $R^t_j$, similarly.
We assume that the following fact holds for any such paths $R^t_i$ and $R^t_j$.

  \begin{assumption}
   \label{assump:stay-half}
   If $x(e_i)$ is half-integral and $x(e_j)$ is an integer,
   and if exactly one of $x(e'_i)$ and $x(e'_j)$ is half-integral,
   then $x(e'_i)$ is half-integral.
  \end{assumption}

Indeed, if Assumption~\ref{assump:stay-half} does not hold, then 
exchanging the subpaths between $v$ and $t$ makes them satisfy it.

In the following discussion, we consider a maximum flow between a terminal $t'$ and $T \setminus \{t'\}$ in $G$, where
$t'$ may equal $t$. In such a flow, each edge $e$ is capacitated by $x(e)$, and each node $v \in V\setminus T$
is assigned the unit capacity or an unbounded capacity if $h=\f$ or $h=\g$, respectively. The capacities of the
terminals are assumed as unbounded. The flow quantity for each $t'$ is at least $r(t')$ if and only if $x$
satisfies \eqref{eq:t-cut}. Let $\Sfam$ be a path decomposition of the
flow between $t'$ and $T\setminus \{t'\}$, in which each path in $\Sfam$
accommodates a half unit of flow. Let $\Sfam_{t}$ be the set of paths in $\Sfam$ that contain nodes in $X^+_t$
(recall that $\hat{X}_t$ is the maximal biset in $\Lfam(t)$). Without loss of generality, we can state the
following fact.

\begin{assumption}\label{assump:path-decomposition}
 Each path in $\Sfam_{t}$ ends at $t$.
 For a path $S \in \Sfam_t$, let $S'$ be the subpath of $S$ between $t$ and the nearest node in
 $V\setminus X^+_t$. 
 Then, $\{S' \colon S \in \Sfam_{t}\} \subseteq
 \{R_1^{t},\ldots,R^t_{2r(t)}\}$ holds.
\end{assumption}

If Assumption~\ref{assump:path-decomposition} is 
not satisfied by $\Sfam$, we can modify the flow between $t'$
and $T\setminus \{t'\}$
by replacing the subpaths of those in
$\Sfam_t$ by appropriate paths in $R^t_1,\ldots,R^t_{2r(t)}$, without decreasing the amount of flow.

We say that $x$ is \emph{minimal in $\cut(h,u)$}
if $x \in \cut(h,u)$ and
no $y \in \cut(h,u)$
exists such that $x\neq y$ and $x(e)\geq y(e)$ for any $e \in E$.
Let edge $e'$ be incident to a node in $X^+_t$.
If $x$ is minimal in $\cut(h,u)$,
then $x(e') = |\{i =1,\ldots,2r(t)\colon e' \in E(R^t_i)\}|/2$;
Otherwise,
as $x(e')$ is decreased, it would remain in $\cut(h,u)$.

\begin{lemma}\label{lem.degree}
Suppose that $h=\f$ or $h=\g$, and 
let $x$ be an extreme minimal point in $\cut(h,u)$. Then
$x(\delta(v))$ is an integer for each $v \in V$.
\end{lemma}
 \begin{proof}
  Define $\bar{x}$ and $x'$
from $x$ as in Corollary~\ref{cor.half-integrality},
and define sets $F$ and $\Lfam$ 
for $x'$ and $\cut(h_{\bar{x}},\allone)$
as in Lemma~\ref{lem.terminal-laminar}.
In other words,
$F=\{e \in E \colon x'(e) =1/2\}$,
and $\Lfam$ is a maximal laminar subfamily of 
$\{\hat{X} \in \Vfam \colon x'(\delta(\hat{X}))=
h_{\bar{x}}(\hat{X})>0\}$
(because $x'(e)<1$ for $e \in E$)
such that the vectors in $\{\eta_{F,\hat{X}} \colon \hat{X} \in \Lfam\}$
are linearly independent.
It suffices to show that $|\delta_F(v)|$ is even for each $v \in V$.

Let $v$ be a node with $\delta_{F}(v)\neq \emptyset$. 
We first observe that $v$ is included by the outer part of some biset in $\Lfam$.
Let $e \in \delta_{F}(v)$.
There exists some $\hat{X}' \in \Lfam$ with $e \in \delta_{F}(\hat{X}')$; otherwise 
a slight decrease in $x$ retains $x$ in $\cut(h,u)$.
Let $\hat{X}$ be the maximal biset such that $\hat{X}'\subseteq \hat{X} \in \Lfam$.
If $v \not\in X^+$, then (ii) of Lemma~\ref{lem.terminal-characterization} implies
the existence of another biset $\hat{Y} \in \Lfam$ with $e \in \delta_{F}(\hat{Y})$,
where $\hat{Y}$ satisfies $v \in Y^+$.

We now prove that $|\delta_F(v)|$ is even. First, we consider the case of $h=\f$.
The laminarity of $\Lfam$ permits two cases: 
(i) the existence of maximal bisets $\hat{X}_1,\ldots,\hat{X}_l \in \Lfam$
with $v \in \Gamma(\hat{X}_1)\cap \cdots \cap \Gamma(\hat{X}_l)$,
and (ii) 
the existence of exactly one maximal biset $\hat{X} \in \Lfam$ with
$v \in X$.

First, we consider the case (i).
In the following discussion,
we show that an even number of edges in $\delta_F(v)$ remains in $G[\hat{X}_i]$
for each $i \in \{1,\ldots,l\}$.
Each edge $e \in \delta_F(v)$ is associated with
exactly one biset $\hat{X}_i$ that includes the both end nodes of $e$ in its outer
part. 
$e$ remains in $G[\hat{X}_i]$, and does not remain in $G[\hat{X}_{i'}]$ 
for any $i' \in \{1,\ldots,l\}$ with $i' \neq i$.
Therefore the claim proves that $|\delta_F(v)|$ is even.
Denote by $t_i$ the terminal with $\hat{{X}}_i \in \Lfam(t_i)$.
Note that
$v$ is included in exactly two paths in $R^{t_i}_1,\ldots,R^{t_i}_{2r(t_i)}$,
say $R^{t_i}_1$ and $R^{t_i}_2$. 
$v$ is adjacent to $s$ in $R^{t_i}_1$ and $R^{t_i}_2$. 
For $j \in \{1,2\}$, let $e_j$ be the edge that joins $v$ to the neighbor opposite $s$ in $R^{t_i}_j$. 
If $e_1=e_2$, then $x(e_1)=1$, and $v$ has no incident edge in $F$ remaining in $G[\hat{X}_i]$.
If $e_1 \neq e_2$, then $x(e_1)=x(e_2)=1/2$. Among the edges in $F$ remaining in $G[\hat{X}_i]$,
these edges alone are incident to $v$.
Hence, the number of edges in $F$ remaining in $G[\hat{X}_i]$
is zero or two.

We now discuss case (ii).
Let $t$ be the terminal with $\hat{X} \in \Lfam(t)$.
By laminarity of $\Lfam$, no biset 
in $\Lfam\setminus \Lfam(t)$ includes $v$ in its outer part. Hence, 
it suffices to show that an even number of edges in 
$\delta_F(v)$ remains in $G^s[\hat{X}]$.
At most two paths in $R^t_1,\ldots,R^t_{2r(t)}$ pass through $v$, but 
if no biset in $\Lfam(t)$ includes $v$ in its neighbor,
$v$ may not be used to full capacity.
However, each edge in $\delta(v)$ is used to full capacity by the minimality of $x$.
If $v \neq t$, then $x(\delta(v))= |\{i \colon v \in V(R^t_{i})\}|$, and $x(\delta(v))$
is an integer. If $v=t$, then $x(\delta(v))=r(t)$, and $x(\delta(v))$ is again an integer.
In either case, $|\delta_F(v)|$ is even, which completes the proof for $h=\f$.

The lemma can be similarly proven for $h=\g$.
Case (i) does not occur because $\Gamma(\hat{X})=\emptyset$
for each $\hat{X} \in \Lfam$.
 \end{proof}

\section{$4/3$-approximation algorithm for the generalized terminal backup problem}
\label{sec.algorithm}

In this section, we prove Theorem~\ref{thm:main-4/3} by presenting
a $4/3$-approximation algorithm for the generalized terminal backup
problem.
We first explain how our algorithm works for the case of $r\equiv 1$ 
for smooth understanding.
Then, we present a full proof of Theorem~\ref{thm:main-4/3}.

\subsection{Algorithm for case of $r\equiv 1$}

Our algorithm rounds a half-integral optimal solution to the LP
relaxations into an integer solution.
Let us assume that a minimal half-integral optimal solution $x$ and a
laminar biset family $\Lfam$ in Lemma~\ref{lem.terminal-laminar}
are given.
In what follows, we explain how to round $x$.

When $r\equiv 1$, the edge- and node-connectivity are equivalent.
Since the neighbor of each biset in $\Lfam$ is empty,
we identify $\Lfam$ with a family of subsets of $V$.

Let $F$ denote $\{e \in E \colon x(e)=1/2\}$.
We call the edges in $F$ \emph{half-integral edges}.
$|\delta_F(v)|$ is even for each $v \in V$
because $x(\delta(v))$ is an integer by Lemma~\ref{lem.degree}.
Hence $F$ can be decomposed into an edge-disjoint set of cycles.
Let $H$ be a cycle in the decomposition.

For each $e \in F$, $\Lfam$ contains a node set to which $e$ is
incident.
Let $\Lfam'$ be the subset of $\Lfam$ that consists of the node sets
to which edges in $H$ are incident.
Since $r \equiv 1$,
 exactly two edges in $H$ are incident to
each node set in $\Lfam'$.

Let $t_1,\ldots,t_k$ be the terminals such that $\Lfam(t_i) \cap \Lfam'
\neq \emptyset$ for each $i \in \{1,\ldots,k\}$.
We can prove that $k$ is an odd number larger than one.
For each $i \in \{1,\ldots,k\}$,
let $X_i$ denote the maximal node set in
$\Lfam(t_i) \cap \Lfam'$,
and let $H_i$ be the subpath of $H$ comprising of edges incident to node
sets in $\Lfam(t_i) \cap \Lfam'$.
If an edge is incident to both $X_i$ and $X_j$, the edge is
shared by $H_i$ and $H_j$.

Let $e_1=uv \in F$ be an edge incident to $X_1$, where we assume without loss
of generality that $u \in X_1$ and  $v \not\in X_1$.
Consider traversing $E(H)$, starting from $e_1$ in the direction from
$v$ to $u$. 
We say that $t_i$ \emph{appears} when
we traverse an edge incident to two node sets $X_i \in \Lfam(t_i)$ and $X_j
\in \Lfam(t_j)$ with $i\neq j$ in the direction from the end node in
$X_j$ to the one in $X_i$.
Without loss of generality, we assume that the terminals appear in the increasing order
of subscripts.
Therefore, during the traverse of $H$, we first visit edges in $H_1$,
then those in $H_2$, and so on.
Suppose that $X\in \Lfam(t_i)$ and $e \in \delta_H(X)$.
We say that $e$ is \emph{outward} with respect to $t_i$ if $e$ is
traversed from the end node in $X$ to the other.
Otherwise, $e$ is called \emph{inward}.
This implies that, during the traverse of $H_i$,
we first traverse edges inward with respect to $t_i$, and then those
outward with respect to $t_i$.

We define $k$ assignments of labels to the edges in $H$, 
where each edge is labeled by either ``$+$'' or ``$-$.''
Let us define the $i$-th assignment.
If $e \in E(H_i)$, then $e$ is labeled by ``$+$.'' If $e \in E(H_j)$ for
some $j > i$, then its label is decided by the following rules.
\begin{itemize}
\item If $j-i$ is odd and $e$ is outward with respect to $t_j$,
$e$ is labeled by ``$-$.''
\item If $j-i$ is odd and $e$ is inward with respect to $t_j$,
$e$ is labeled by``$+$.''
\item If $j-i$ is even and $e$ is outward with respect to $t_j$,
$e$ is labeled by ``$+$.''
\item If $j-i$ is even and $e$ is inward with respect to $t_j$, 
$e$ is labeled by``$-$.''
\end{itemize}
If $e \in E(H_j)$ for some $j < i$, we assign the opposite label to the
above rules; For example, if
$i-j$ is odd and $e$ is outward with respect to $t_j$,
$e$ is labeled by ``$+$.''

Note that this assignment is consistent;
if
$e$ is included in both $H_j$ and $H_{j+1}$, then
$e$ is outward with respect to $t_i$ and inward with respect to
$t_{i+1}$,
and hence $e$ is assigned the same label from $j$ and $j+1$;
$e_1$ is shared by $H_1$ and $H_k$, and similarly it is assigned the same label 
because $k$ is odd.
Figure~\ref{fig.cicle2} shows an example of the cycle $H$, and the first
assignment of labels to the edges on $H$.

\begin{figure}
 \centering
 \includegraphics[scale=.8]{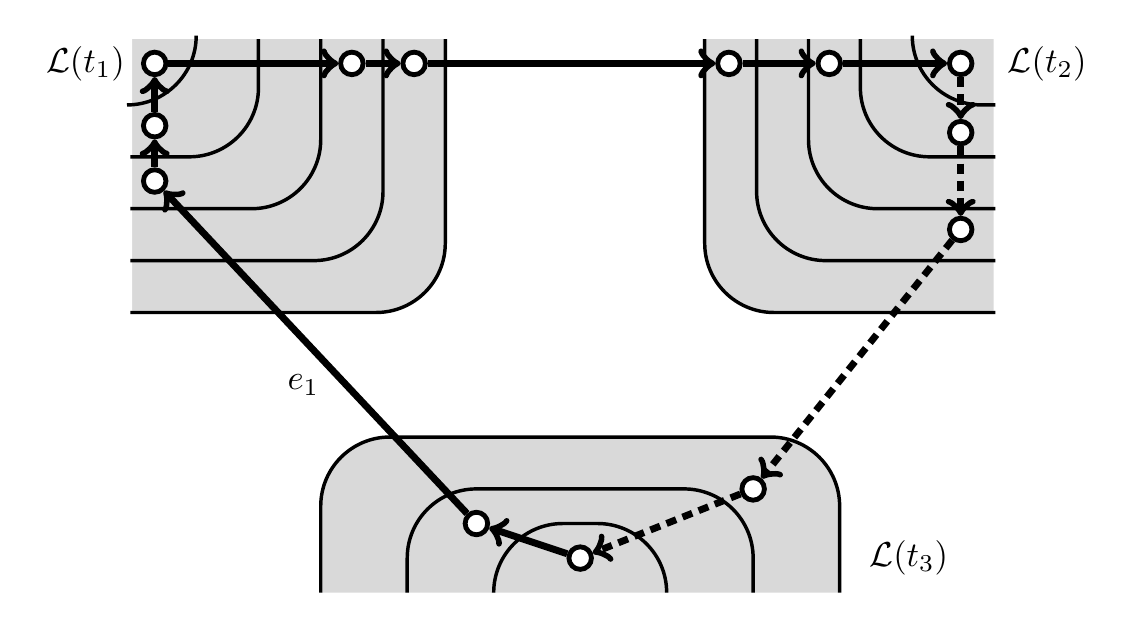}
\caption{An example of a cycle of half-integral edges and the first
 assignment of labels to the edges. 
Edges drawn by 
solid and dashed lines are assigned ``$+$'' and ``$-$,'' respectively.
The edges are oriented in the direction of traverse. The areas surrounded by thin solid lines
represent the node sets in $\Lfam$.}
\label{fig.cicle2}
\end{figure}

Our algorithm rounds $x(e)$ into $1$ if $e$ is labeled by ``$+$,''
and into $0$ otherwise. 
Since we have $k$ assignments of labels, we have $k$ ways of rounding of
variables corresponding to the edges in $H$.
Our algorithm chooses the most cost-effective one among them.

Let us observe that this algorithm is $4/3$-approximation. First, we
prove that the above rounding increases the cost by a factor of at most $4/3$.
Let $x'$ be the vector obtained from $x$ by the rounding.

\begin{lemma}\label{lem.cost-terminal}
\[
\sum_{e \in E}c(e)x'(e) \leq \frac{4}{3} \sum_{e \in E}c(e)x(e).
\]
\end{lemma}
\begin{proof}
 Let $H$ be a cycle of half-integral edges.
 We show that $\sum_{e \in H}c(e)x'(e) \leq \frac{4}{3} \sum_{e \in H}c(e)x(e)$.
 Applying this claim to all cycles in the decomposition of $F$, we
 can prove the lemma.
 We use the notations used in the definition of the rounding.
 
 Let $x_{i}$ denote the vector obtained by rounding $x(e)$, $e \in E(H)$
according to the $i$-th assignment of labels.
We note that 
\[\sum_{e \in H}c(e)x'(e)= \min_{1\leq i \leq k} \sum_{e \in H}c(e)x_{i}(e)  \leq \frac{1}{k} \sum_{i=1}^k\sum_{e \in H}c(e)x_{i}(e).
\] 

Recall that $k$ is an odd number larger than one.
In the $k$ assignments, $e \in H$ is 
labeled ``$+$'' by the $(k+1)/2$ assignments.
Thus, 
\[
\sum_{i=1}^k\sum_{e \in H}c(e)x_{i}(e) =\frac{k+1}{2} \sum_{e \in H}c(e).
\]
Note that $\sum_{e \in H}c(e)x(e)=\sum_{e \in H}c(e)/2$. Therefore,
\[
	\frac{\sum_{e \in H}c(e)x'(e)}{\sum_{e \in H}c(e)x(e)}
	\leq
	\frac{k+1}{k} \leq \frac{4}{3},
\]
where the last inequality follows from $k\geq 3$.
\end{proof}

Next, let us prove the feasibility of $x'$.
For a path $P$ and nodes $u,v$ on $P$, we denote the subpath of $P$ between
$u$ and $v$ by $P[u,v]$.

\begin{lemma}\label{lem.feasibility-easy}
 $x'$ is a feasible solution to the terminal backup problem.
\end{lemma}
 \begin{proof}
  Obviously $x'$ is an integer vector. Hence, to prove the feasibility
  of $x'$, the graph with edge-capacities $x'$
  admits a unit of flow from each terminal $t$ to the other terminals.
  Since $x(\delta(X)) \geq 1$ for each $X \in \Cfam(t)$,
  the graph capacitated by $x$ admits such a flow.
  Hence we show that a flow for $x'$ can be obtained by
  modifying the flow for $x$.
  In the following, we assume that $x'$ is obtained by rounding variables
  corresponding to the half-integral edges in a cycle $H$.
  If required, the modification is repeated for each cycle of half-integral edges.

  Recall the definition of $R^t_1,\ldots,R^t_{2r(t)}$ in Section~\ref{subsec.path-decomp}.
  Since we are considering the case of $r\equiv 1$, we have two paths $R^t_1$ and
  $R^t_2$ for each terminal $t$ with $\Lfam(t)\neq \emptyset$.
  We assume these paths satisfy Assumption~\ref{assump:stay-half}.
  Fix a terminal $t$, and 
  suppose that the flow from $t$ to the other terminals
  with edge-capacities $x$ delivers a half unit of flow along a path $P$,
  and another half unit along a path $Q$.
  We assume that $\Sfam=\{P,Q\}$ satisfies Assumption~\ref{assump:path-decomposition}.

  If both $P$ and $Q$ contains no half-integral edge (with respect to $x$) labeled by
  ``$-$,''
  the flow satisfies the capacity constraints defined from $x'$.
  Thus, let us consider the case where $P$ includes a half-integral edge labeled
  by ``$-$.''
    Let $e$ be the one nearest to $t$ among such edges, and let $v$ be
  the end node of $e$ near to $t$.

  We first show that there exists $X^* \in \Lfam(t)$ such that $e \in
  \delta(X^*)$ and $v \in X^*$. For arriving at a contradiction,
  suppose that such $X^*$ does not exist.
  $e$ is incident to at least one node set in $\Lfam$.
  In particular, Lemma~\ref{lem.terminal-characterization}(ii)
  implies that there exists a terminal $t' \in T$ and node set $X' \in \Lfam(t')$ such that $e \in \delta(X')$ and
  $v \in X'$.
  However, this means that $t'\neq t$ and $P[t,v]$ enters $X'$ when traversed from
  $t$ to $v$.
  Assumption~\ref{assump:path-decomposition} indicates that
  the subpath of $P$ between $v$ and the end opposite to $t$ is included
  by
  $R^{t'}_1$ or $R^{t'}_2$. Hence,
  the end of $P$ opposite to $t$ is $t'$, and $P$ does not include
  $e$, which is a contradiction. Therefore,
  there exists $X^* \in \Lfam(t)$ such that $e \in
  \delta(X^*)$ and $v \in X^*$.

  This fact indicates that $Q$ contains no ``$-$''-labeled half-integral
  edge because of the following reason.
  Let $P'$ be the subpath of $P$ that is included by a maximal node set
  in $\Lfam(t)$.
  Since $\Lfam(t) \neq \emptyset$, there exists $R^t_1$ and $R^t_2$.
  By Assumption~\ref{assump:path-decomposition},
  $P'$ is equal to $R^t_1$ or $R^t_2$. Without loss of generality, let $P'$ be
  equal to $R^t_1$.
  Then, Assumption~\ref{assump:stay-half} indicates that 
  all ``$-$''-labeled half-integral edges incident to node sets in
  $\Lfam(t)$  is included in $R^t_1$.
  Since $P$ and $Q$ share no half-integral edges, $Q$ does not include
  these edges in $R^t_1$.
  Hence, if $Q$ contains a ``$-$''-labeled half-integral edge, its both
  end node is included by some node sets in $\Lfam \setminus \Lfam(t)$.
  However, we can derive a contradiction similarly for the above claim
  with $P$.

  Since $x(\delta(X^*))=1$,
  the other edge $e'$ incident to $v$ on $H$ is also incident to $X^*$.
  By the label-assignment rules, $e'$ is labeled by ``$+$.''
  Let $H'$ denote the subpath of
  $H$ consisting of ``$+$''-labeled edges and terminating at $v$.
  Let $u$ be the other end node of $H'$,
  and let $g$ be the edge incident to $u$ on $H'$.
  By Lemma~\ref{lem.terminal-characterization}, there exists $Y \in
  \Lfam$ with
  $g \in \delta(Y)$ and  $u \in Y$.
  $Y$ belongs to $\Lfam(t')$ for some $t' \neq t$.
  $g$ is included in a path $R_1^{t'}$ or $R_{2}^{t'}$.
  Without loss of generality, we suppose that $R_1^{t'}$ includes $g$.
  We replace $P$ by the concatenate of $P[t,v]$, $H'$, and $R_1^{t'}[u,t']$.
  See Figure~\ref{fig.modification} for illustration of this modification.
  
  \begin{figure}
  \centering
   \includegraphics[scale=.8]{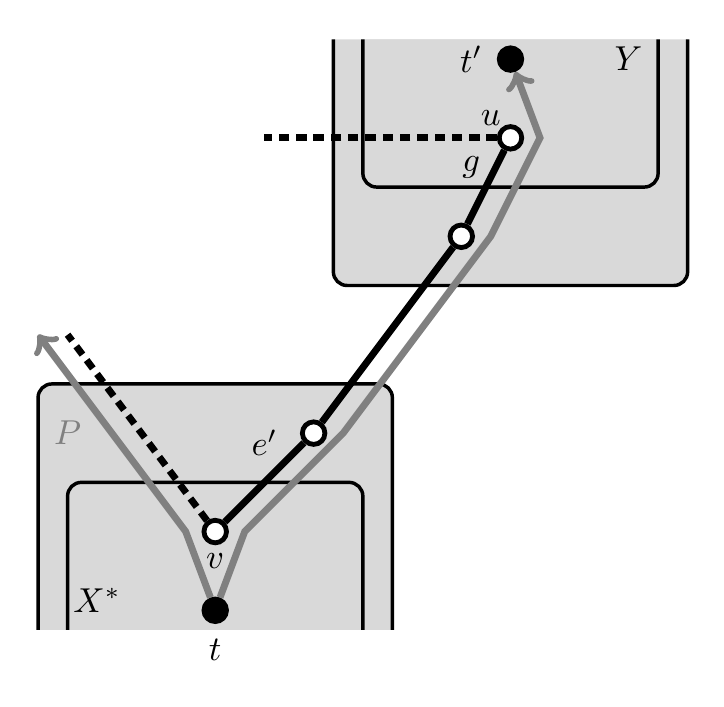}
   \caption{The definitions in the proof of Lemma~\ref{lem.feasibility-easy}}
   \label{fig.modification}
  \end{figure}
  
  Let us observe that this modification preserves the capacity
  constraints.
  $P[t,v]$ was a part of $P$ before the modification.
  The capacity of each edge on $H'$ is increased by $1/2$ when $x'$
  replaces $x$.
  The capacity of each edge in $R_1^{t'}[u,t']$ is integer.
  Hence no capacity constraint is violated.
 \end{proof}

\subsection{Algorithm for the general case}
In this subsection, we present a strongly polynomial-time algorithm for the generalized terminal
backup problem.
In the following discussion,
$h$ denotes a skew supermodular function such that $h(\hat{X})>0$
only when $\hat{X} \in \Cfam$.

\subsubsection*{Solving the LP relaxation}
We wish to ensure that any optimal solution $x$ to $\lp(h,u)$ 
is minimal in $\lp(h,u)$. Clearly, this condition holds when $c(e) > 0$ for each $e \in E$.
If $c(e)=0$ for some $e \in E$, the condition is ensured by perturbing $c$.
Since we can restrict our attention to half-integral solutions,
it is sufficient to reset $c(e)$ 
to a positive number smaller than $2/(\theta|E|)$ for each $e$ with $c(e)=0$,
where $\theta$ is the maximum denominator of the edge costs.

The number of constraints of $\lp(h,u)$ is exponential;
hence, it is unclear how to solve $\lp(h,u)$ in polynomial time.
If $h=\f$ or $h=\g$,
the separation is reducible to
a maximum flow computation, and $\lp(h,u)$ can be solved by the ellipsoid method.
Alternatively, the constraints can be written in a compact form by introducing flow variables for
each terminal, as implemented in Jain~\cite{Jain01}.
Hence, if $h=\f$ or $h=\g$, there are two ways of solving
$\lp(h,u)$ in polynomial time.
However, Theorem~\ref{thm:main-4/3} claims a strongly polynomial-time algorithm.
All coefficients in the constraints of $\lp(h,u)$ are one. Accordingly, 
Tardos' algorithm~\cite{Tardos1986} computes
an optimal solution to $\lp(h,u)$ in strongly polynomial time,
but does not guarantee an extreme point solution.

Our algorithm first finds an optimal solution to $\lp(h,u)$ by Tardos' algorithm.
The obtained solution is denoted by $x^*$.
Defining $\bar{x}^*\colon E \rightarrow \Zset_+$ by $\bar{x}^*(e)=\lfloor x^*(e)\rfloor$ for $e \in
E$, we then compute an extreme point optimal solution $x$ to $\lp(h_{\bar{x}^*},\allone)$.
$\bar{x}^* + x$ is not necessarily an extreme point of $\cut(h,u)$,
but is a half-integral optimal solution to $\lp(h,u)$.
The following lemma shows that 
$x$ can be computed by iterating Tardos' algorithm.

\begin{lemma}\label{lem.lp-stronglypoly}
An extreme point optimal solution to $\lp(h_{\bar{x}^*},\allone)$ 
can be computed in strongly polynomial time.
\end{lemma}
 \begin{proof}
  As noted above, 
an optimal solution to $\lp(h_{\bar{x}^*},\allone)$ 
can be computed in strongly polynomial time.
Moreover, 
whether fixing a variable $x(e)$ to a specific value $\tau$
increases the optimal value is also testable in strongly polynomial time 
by solving $\lp(h_{\bar{x}^*},\allone)$ 
with an additional constraint $x(e)=\tau$. 
We sequentially test fixing the variables $x(e)$ to $0$ or $1$,
and if the fix does not increase the optimal value, the variable is set to the fixed value.
If $x(e)$ is not fixed to $0$ or $1$,
it is set to $1/2$.

Optimality of the above-constructed solution $x$ 
follows from the existence of a half-integral optimal
solution (see Lemma~\ref{lem.terminal-characterization}).
We must now prove that the obtained solution $x$ is an extreme point.
If not, $x$ can be represented by
$\sum_{i=1}^l \alpha_i y_i$, where $l \geq 2$,
$y_1,\ldots y_l$ are extreme points of $\cut(h_{\bar{x}^*},\allone)$,
and $\alpha_1,\ldots,\alpha_l$ are positive real numbers with
$\sum_{i=1}^l \alpha_i=1$.
Let $i \in \{1,\ldots,l\}$.
The optimality of $x$ indicates that
$y_i$ is an optimal solution to 
$\lp(h_{\bar{x}^*},\allone)$.
Moreover,  $y_i(e)=x(e)$ holds if $x(e) \in \{0,1\}$.
Therefore, there exists some $e \in E$ such that $x(e)=1/2$ and $y_i(e) \in \{0,1\}$,
which contradicts the way of constructing $x$.
 \end{proof}

 Let $F=\{e \in E \colon x(e) =1/2\}$.
 Our algorithm also requires $\Lfam$ defined from $x$ (i.e., $\Lfam$ is
 a maximal laminar subfamily of $\{\hat{X} \in \Cfam \colon
 x(\delta_F(X))=h_{\bar{x}^*}(\hat{X})>0\}$
such that the vectors $\eta_{F,\hat{X}}$, $\hat{X} \in \Lfam$
 are linearly independent).
 As stated in the paragraph following Lemma~\ref{lem.terminal-laminar},
 $\Lfam$ can be constructed by
 repeatedly adding a
 biset $\hat{Y}$ in $\{\hat{X} \in \Cfam \colon x(\delta_F(X))=h_{\bar{x}^*}(\hat{X})>0\}$
 such that adding $\hat{Y}$ to $\Lfam$ preserves the laminarity of
 $\Lfam$
 and the linear independence of the vectors 
 $\eta_{F,\hat{X}}$, $\hat{X} \in \Lfam$.
 If $\Lfam$ is not maximal, such a biset $\hat{Y}$ 
 can be found as follows.
 By Lemma~\ref{lem.terminal-characterization}, 
 one of such $\hat{Y}$ satisfies either of the following conditions:
 \begin{itemize}
  \item[(i)]  $\hat{Y}$ is minimal in $\{\hat{Y}\} \cup \Lfam$, and $|\delta_F(\hat{Y})|=2$;
  \item[(ii)] There exits $\hat{X} \in \Lfam$ such that
	$\hat{X} \subseteq \hat{Y}$ and
	$|\delta_F(\hat{X})\setminus \delta_F(\hat{Y})|+|\delta_F(\hat{Y})\setminus \delta_F(\hat{X})|=2$.
 \end{itemize}
 The number of bisets satisfying one of these conditions is strongly
 polynomial.
We can decide in strongly polynomial time whether adding
a biset to the current $\Lfam$ preserves the conditions of $\Lfam$. 
Therefore, $\Lfam$ can computed in strongly polynomial time.

\subsubsection*{Rounding half-integral solutions to $4/3$-approximate solutions}
Our algorithm rounds $x$, the extreme point optimal solution
to $\lp(h_{\bar{x}^*},\allone)$,
to an integer vector $x' \in \cut(h_{\bar{x}^*},\allone)$
subject to $\sum_{e \in E}c(e)x'(e) \leq 4/3 \cdot \sum_{e \in E}c(e)x(e)$.
It then outputs $\bar{x}^* + x'$.

The rounding procedure
is almost same as the algorithm for $r\equiv 1$.
Let $F=\{e \in E \colon x(e)=1/2\}$.
By Lemma~\ref{lem.degree}, $|\delta_F(v)|$ is even for each $v \in V$ because
$\bar{x}^* + x$ is minimal in $\cut(h,u)$.
We can see that $|\delta_F(v)|$ is an even number at most four.

\begin{lemma}
 $|\delta_F(v)| \leq 4$ for each $v \in V$.
 If $|\delta_F(v)| = 4$, there exist $\hat{X},\hat{W} \in \Lfam$
 such that $\hat{W} \subseteq \hat{X}$, $v \in X \setminus W^+$,
 and $|\delta_F(v)\cap \delta_F(\hat{X})|=2=|\delta_F(v)\cap \delta_F(\hat{W})|$.
\end{lemma}
\begin{proof}
 Let $\delta_F(v) \neq \emptyset$. Then,
 Lemma~\ref{lem.terminal-characterization} (ii) implies that
 $v$ is included in the inner-part of some biset in $\Lfam$.
 Let $\hat{X}$ be the minimal biset in $\Lfam$ such that $v \in X$.
 If $\hat{X}$ is minimal in $\Lfam$,
 then $\delta_F(v) \subseteq \delta_F(\hat{X})$, and
 $|\delta_F(v)|\leq 2$ follows from $|\delta_F(\hat{X})|= 2$.
 In the rest of the proof, suppose that
 $\hat{X}$  has the child $\hat{Y} \in \Lfam$.
 Then $v \not\in Y$.
 Suppose that $v \in Y^+$, and let $\hat{Z}$ be the minimal biset in
 $\Lfam$
 such that $\hat{Z} \subseteq \hat{Y}$ and $v \in Z^+$,
 where $\hat{Z}$ is
 possibly equal to $\hat{Y}$.
 Let $\hat{W}$ be the child of $\hat{Z}$.
 Each edge  $e \in \delta_F(v)$ is incident to $\hat{X}$ or $\hat{W}$.
 Notice that $e$ is not incident to $\hat{Y}$ or $\hat{Z}$.
 Hence,
 $e \in F^+_{\Lfam}(\hat{X})$ if $e$ is incident to $\hat{X}$, 
 and $e \in F^-_{\Lfam}(\hat{Z})$ if $e$ is incident to $\hat{W}$.
 Thus $|\delta_F(v)| \leq |F^+_{\Lfam}(\hat{X})| + |F^-_{\Lfam}(\hat{Z})| \leq 4$.
 If $\hat{W}$ does not exist,
 $|\delta_F(v)| \leq |F^+_{\Lfam}(\hat{X})| \leq 2$.
 If $v \not\in Y^+$, we can similarly show that $\delta_F(v) \subseteq
 F^+_{\Lfam}(\hat{X}) \cup F^-_{\Lfam}(\hat{X})$, and hence
 $|\delta_F(v)| \leq 2$ follows from
 $|F^+_{\Lfam}(\hat{X})| + |F^-_{\Lfam}(\hat{X})|\leq 2$.
\end{proof}

We decompose $F$ into a set of cycles.
We assume without loss of generality that the decomposition satisfies
the following assumption.

 \begin{assumption}
  \label{assump.cycle}
 Let $v$ be a node such that $|\delta_F(v)|=4$. Let $\hat{X},\hat{W} \in
 \Lfam$ be the bisets such that $\hat{W} \subseteq \hat{X}$, $v \in X
  \setminus W^+$, and
 $|\delta_F(v)\cap \delta_F(\hat{X})|=2=|\delta_F(v)\cap
  \delta_F(\hat{W})|$.
  Then the two edges in $\delta_F(v)\cap \delta_F(\hat{X})$
  {\rm (}resp., $\delta_F(v)\cap \delta_F(\hat{W})${\rm )}
  are included in the same cycle in the decomposition.
   \end{assumption}

  Suppose that $H$ includes an edge incident to a biset in
  $\Lfam(t_1)$ and another in $\Lfam(t_k)$ for some terminals
  $t_1,t_k$ with $t_1 \neq t_k$.
   Let $e_1$ be one of such edges.
   We traverse $H$, starting from $e_1$.
   Suppose that $e_1$ is traversed from a biset in $\Lfam(t_k)$ to one in
   $\Lfam(t_1)$.
   Let $(t_1,\ldots,t_k)$ be the
   sequence of terminals that appear when we traverse $H$ from $e_1$,
   where $t_i$ denotes the terminal that appears immediately after $t_{i-1}$.
   A different fact from the case of $r\equiv 1$ is that
   a terminal can appear more than once during the traverse.
   Thus $t_i$ and $t_j$ may stand for the same terminal unless $j \in
   \{i-1,i+1\}$ or $\{i,j\}=\{1,k\}$.

\begin{figure}
\centering
\includegraphics[scale=.6]{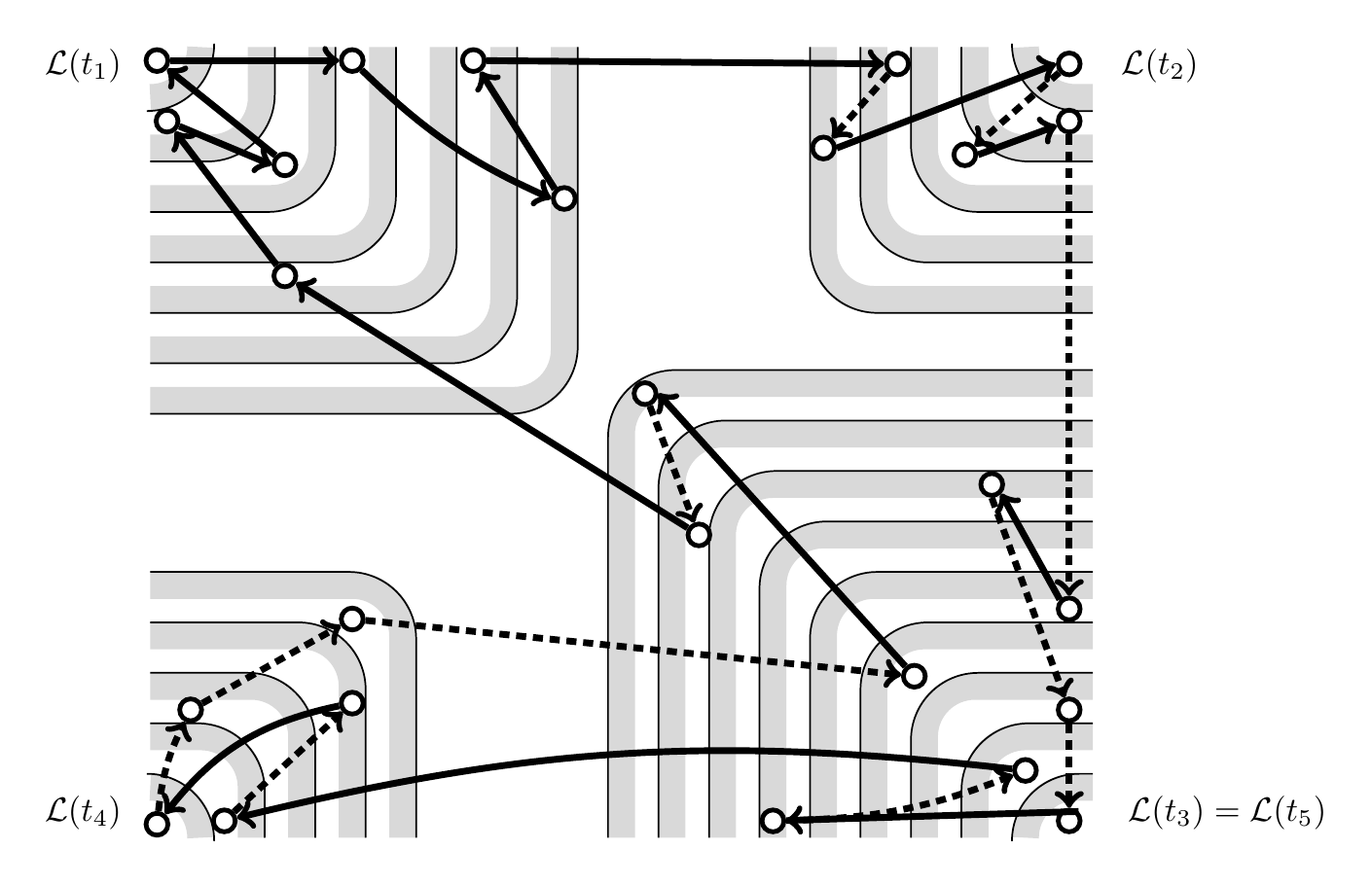}
\caption{An example of a cycle of half-integral edges and labels assigned to the edges. 
Edges drawn by 
solid and dashed lines are assigned ``$+$'' and ``$-$,'' respectively.
The edges are oriented in the direction of traverse. The areas surrounded by thin solid lines
represent the outer parts of bisets in $\Lfam$, and gray areas indicate their neighbors.
In this figure, neighbors of bisets in $\Lfam$ are disjoint for visibility, but neighbors can overlap
in general.}
\label{fig.cicle}
\end{figure}

Let $H_i$ be the subpath of $H$ that consists of edges
between the appearance of $t_i$ and $t_{i+1}$, where
$H_i$ and $H_{i+1}$ share an edge that is incident to both a biset in
$\Lfam(t_i)$ and one in $\Lfam(t_{i+1})$,
and $H_1$ and $H_k$ share $e_1$.
We also define ``inward'' and ``outward'' edges in $H_i$ with respect to
$t_i$
as in the case of $r \equiv 1$.
Another different fact in the general case from the case of $r\equiv 1$ is
that the direction of edges on $H_i$ with respect to $t_i$ changes more
than once because $H$ may contain more than two edges incident to a biset in $\Lfam$.

If all edges on $H$ are incident to only bisets in $\Lfam(t)$ for some
terminal $t$, we let $k=1$, and $H_1=H$ for convention.
In the following lemma, we see that $k$ is an odd number larger than one.

\begin{lemma}\label{lem.even}
A cycle such that $k$ is one or an even number does not exist.
\end{lemma}
\begin{proof}
Suppose that $k$ is one or an even number for a cycle $H$.
 Let us assign labels to each edge in $H$ as follows.
 Let $e \in H_i$.
 If $i$ is odd and $e$ is inward to $t_i$, or if
 $i$ is even and $e$ is outward to $t_i$, 
 then $e$ is labeled ``$-$.''
 Otherwise, $e$ is labeled ``$+$.''
 We note that, for each $\hat{X} \in \Lfam$,
exactly half of the edges  in $\delta_H(\hat{X})$
are labeled by ``$+$.''

 Let $\epsilon$ be a constant. 
 For each edge $e$ in $H$,
 update the corresponding variable $x(e)$ to $x(e)+\epsilon$
 if $e$ is labeled by ``$+$'',
 and update to $x(e)-\epsilon$ otherwise.
Let $x_{\epsilon}$ denote the obtained vector.
The number of labels assigned indicates that 
 $x_{\epsilon}(\delta_F(\hat{X}))=x_{-\epsilon}(\delta_F(\hat{X}))=h_{\bar{x}^*}(\hat{X})$
 for each $\hat{X} \in \Lfam$.
 If
 $x(\delta_F(\hat{Y}))=h_{\bar{x}^*}(\hat{Y})$ holds for a biset
 $\hat{Y} \not\in \Lfam$,
 $x_{\epsilon}(\delta_F(\hat{Y}))=x_{-\epsilon}(\delta_F(\hat{Y}))=h_{\bar{x}^*}(\hat{Y})$
is implied by the linear dependence of $\eta_{F,\hat{Y}}$  from
 $\eta_{F,\hat{X}}$, $\hat{X} \in \Lfam$, shown in Lemma~\ref{lem.terminal-laminar}.
Therefore, both $x_{\epsilon}$ and $x_{-\epsilon}$
 belong to 
$\cut(h_{\bar{x}^*},\allone)$ for a sufficiently small positive number
$\epsilon$,
contradicting that $x$ is an extreme point of $\cut(h_{\bar{x}^*},\allone)$.
\end{proof}

Since $k\geq 3$ by Lemma~\ref{lem.even}, we can choose $t_1$ so that
$t_2 \neq t_k$. We assume this condition in the rest of this section.

We define $k$ assignments of labels ``$+$'' and ``$-$''
to the edges on $H$ as in the case of $r\equiv 1$.
 Figure~\ref{fig.cicle} illustrates 
 a cycle of half-integral edges and the first assignment of labels 
 to its edges. In this example, $k=5$, and $t_3$ and $t_5$ indicate the same terminal.

Our algorithm computes an integer vector $x'$ from $x$ as follows.
For each cycle $H$ of half-integral edges,
the algorithm selects the most cost-effective choice from $k$ assignments of labels.
Based on the labels,
$x$ is rounded to obtain the vector $x'$;
If an edge $e$ is labeled by ``$+$'',  $x'(e)$ is defined as $1$.
Otherwise, $x'(e)$ is $0$.
Recall that the algorithm outputs $\bar{x}^* + x'$.

\subsubsection*{Performance guarantee}

We can prove
$\sum_{e \in E}c(e)x'(e) \leq 4/3 \cdot \sum_{e \in E}c(e)x(e)$
similarly for Lemma~\ref{lem.cost-terminal}.
The next lemma proves that $\bar{x}^*+x'$ is a feasible solution.
Theorem~\ref{thm:main-4/3} is immediately proven from
these facts and Lemmas~\ref{lem.lp-stronglypoly}.

\begin{lemma}\label{lem.feasibility}
 $x' \in \cut(h_{\bar{x}^*},\allone)$ when $h=\f$ or $h=\g$.
\end{lemma}
\begin{proof}
Consider the case of $h=\f$.
Assume that 
nodes in $V\setminus T$ have unit capacities
and nodes in $T$ have unbounded capacities.
We also regard 
$\bar{x}^*+x$ and $\bar{x}^*+x'$ as edge capacities.
To prove that $x' \in \cut(h_{\bar{x}^*},\allone)$, 
it suffices to show that, for each $t \in T$,
the graph capacitated by $\bar{x}^*+x'$ admits
a flow of amount $r(t)$
between $t$ and $T\setminus \{t\}$.

Now consider a maximum flow between $t$ and $T\setminus \{t\}$
in the graph capacitated by $\bar{x}^*+x$.
Suppose that the maximum flow is decomposed into a set 
$\Sfam$ of paths, each running a half unit of flow from $t$ to another terminal.
Since $x$ satisfies $x(\delta(\hat{X})) \geq \f_{\bar{x}^*}(\hat{X})$
for each $\hat{X}\in \Vfam$,
the flow amount is at least $r(t)$ (i.e., $|\Sfam|\geq 2r(t)$).
Recall that we are assuming Assumption~\ref{assump:path-decomposition}.
We now modify $\Sfam$ to satisfy the capacity constraints when 
the capacity of $e \in E$ is changed from $\bar{x}^*(e) + x(e)$ to $\bar{x}^*(e) + x'(e)$.
In the following, we assume that $x'$ is obtained by rounding variables corresponding to 
the half-integral edges in a cycle $H$. If required, the modification is repeated for each cycle
of half-integral edges.
We define the notations such as $t_1,\ldots,t_{k}$ and $H_1,\ldots,H_k$
 from $H$ as we defined above.

We traverse $S \in \Sfam$ from $t$ to the other end.
When arriving at an edge $e \in E(H)$ labeled by ``$-$,''
we reroute the flow along $S$ as follows.
Let $v$ be the end node of $e$ near to $t$.
By Assumption~\ref{assump:path-decomposition} and the label-assignment rules,
$e$ shares node $v$ with an edge labeled ``$+$'' on $H$.
Let $H'$ denote the subpath of $H$ consisting of ``$+$''-labeled edges and terminating at $v$.
We follow $H'$ instead of $e$.
Let $u$ be the other end node of $H'$,
and let $e'$ be the edge incident to $u$ on $H'$.
By Lemma~\ref{lem.terminal-characterization}, there exists $\hat{X}\in \Lfam$ with $u \in X^+$.

\begin{figure}
\centering
\includegraphics[scale=.7]{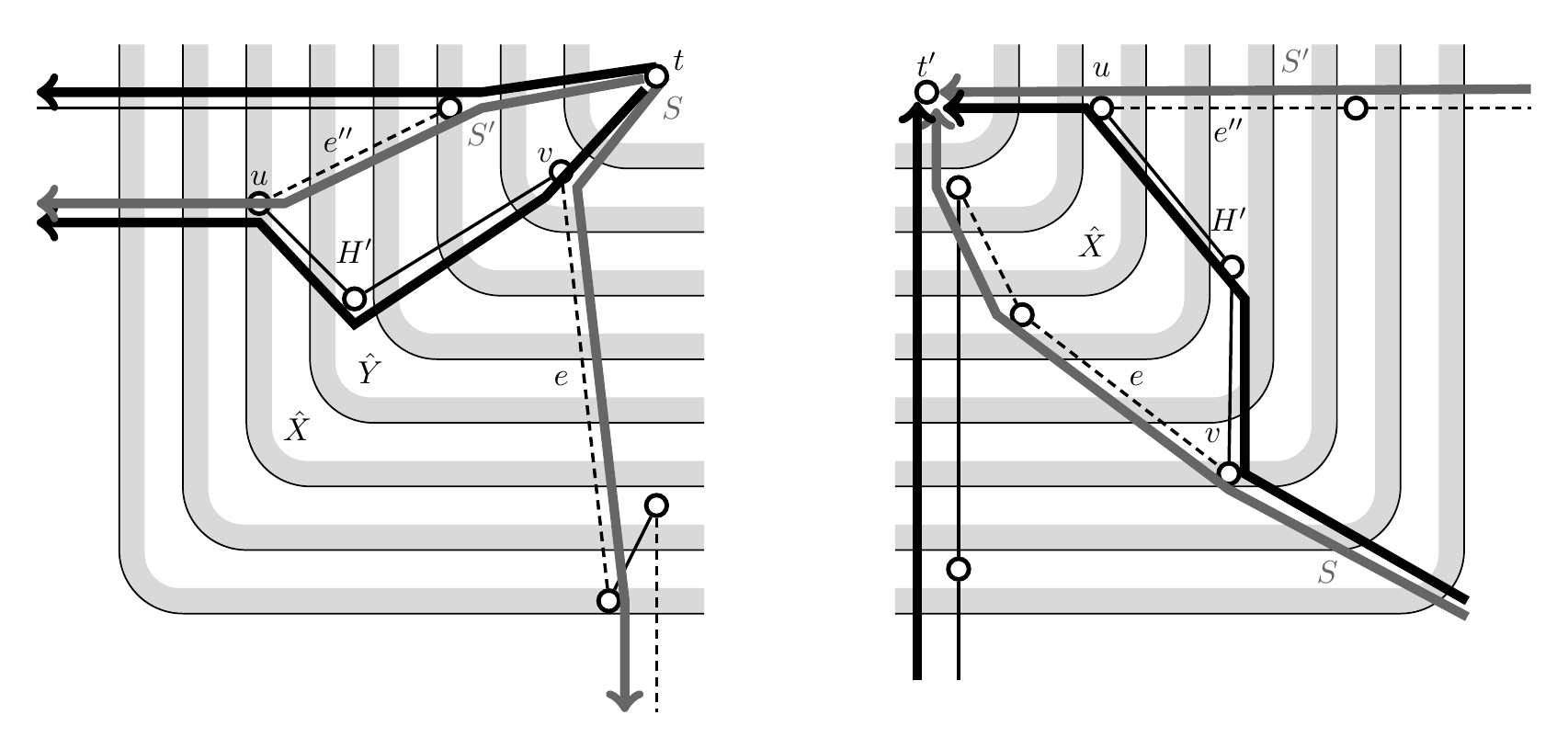}
\caption{Transformation of $\Sfam$ in the proof of Lemma~\ref{lem.feasibility}. The left and right
panels illustrate the cases of $\hat{X}\in \Lfam(t)$ and $\hat{X} \in
\Lfam(t')$, respectively, with $t \neq t'$. The paths $S$ and $S'$ are represented by dark gray lines;
the black lines represent the paths obtained by modifying $S$ and $S'$.}
\label{fig.flow-transformation}
\end{figure}

Suppose that $\hat{X} \in \Lfam(t)$. 
Let $\hat{X}$ be the minimal biset such that $\hat{X} \in \Lfam(t)$ and $u \in X^+$, and
let $\hat{Y}$ be the child of $\hat{X}$.
Then, $e' \in \delta(\hat{Y})$, and $u \in X^+ \setminus Y^+$.
Moreover, another half-integral edge $e'' \in \delta_H(\hat{Y})$, labeled ``$-$,'' is incident to
$u$.
Edge $e''$ is included in another path $S' \in \Sfam$.
Let $t'$ be the terminal such that $t\neq t'$ and $S' \in \Sfam_{t'}$.
After reaching $u$, we move to $t'$ along the path $S'$.
In other words,
path $S$ is replaced by the concatenate of $S[t,v]$,
$H'$, and $S'[u,t']$.
If $S'[u,t']$ contains a half-integral edge labeled by ``$-$'', we modified it recursively.
These definitions are illustrated in the left panel of Figure~\ref{fig.flow-transformation}.
Let us observe that this modification does not violate 
the capacity constraints when the edges are capacitated by $\bar{x}^* + x'$.
 Assumption~\ref{assump.cycle}
 indicates that
 exactly two half-integral edges are incident to each inner node on $H'$.
 The capacity of each edge on $H'$ increases by $1/2$ by the modification,
 exactly counterbalancing the unused half capacity
of each inner node on $H'$ prior to the modification.
 Even if the capacities of edges and nodes on
 $S'[u,t']$ are used before the modification,
 the flow along $S'$ is modified so that these capacities are unused.
 Thus the capacity constraints are preserved by the
 modification.

 Next, suppose that $\hat{X} \in \Lfam(t')$ for some $t'$ with $t\neq t'$.
 We first consider the case of $t' \neq t_1$.
 Lemma~\ref{lem.terminal-characterization}(ii) indicates that
 we can assume $u \in X$.
 We let $\hat{X}$ be the minimal among such bisets.
Another half-integral edge $e'' \in \delta(\hat{X})$, labeled by ``$-$,'' is incident to $u$,
and is included in a path in $\{R_1^{t'},\ldots,  R_{2r(t')}^{t'}\}$.
Without loss of generality, we suppose that that $R_1^{t'}$ is such a path.
After arriving at $u$, we reach $t'$ along $R_1^{t'}[u,t']$, 
as shown in the right panel of Figure~\ref{fig.flow-transformation}.
Again, this modification preserves the capacity constraints. To see this,
suppose that another path $S' \in \Sfam \setminus \{S\}$ includes $R_1^{t'}$.
Then,
$S'$ includes a ``$-$''-labeled edge before reaching $u$ when traversed from $t$ to $t'$.
$S'$ will be diverted to another route, and 
half of the edge and node capacity on $R_1^{t'}[u,t']$ will be no longer used.
Prior to modification,
half of the inner node capacity of $H'$ was unused because
the nodes were incident to exactly two half-integral edges.

 We next discuss the case of $t'=t_1$.
 In this case, $t \in \{t_2,t_{k}\}$.
Recall that all edges in $H_1$ are labeled by ``$+$.''
Each of $H_2$ and $H_k$ shares exactly one edge with $H_1$. We let $e_1$
 denote the edge shared by $H_1$ and $H_k$, and $e_2$ denote the one
 shared by $H_1$ and $H_2$.
$e_1$ and $e_2$ are traversed inward and outward with respect to $t_1$, respectively.
If $t=t_2$, we modify each path in $\Sfam_{t_1}$
as when each outward-traversed edge in $H_1$ is labeled ``$-$,'' whereas other edges are
labeled ``$+$.''
If $t=t_k$, we perform the converse operation, implemented when
each outward-traversed edge in $H_1$ is labeled ``$+$,'' whereas other edges are labeled ``$-$.''
 The modification when $t=t_2$ is illustrated in Figure~\ref{fig.flow2}.
 Recall that we chosen $t_1$ so that $t_2 \neq t_k$.
The capacity constraints are preserved because
no path in $\Sfam$ includes $e_1$ when $t=t_2$, and 
no path in $\Sfam$ includes $e_2$ if $t=t_k$
before the modification.

\begin{figure}\centering
\includegraphics{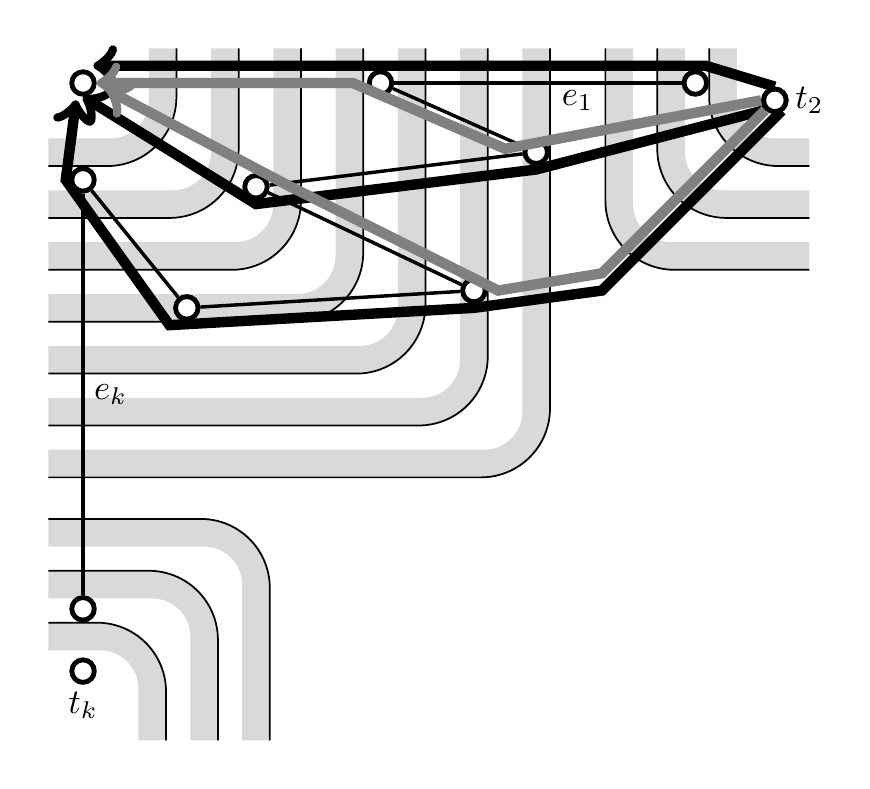}
\caption{Modification of $\Sfam$ when $t=t_2$. Gray thick lines represent paths before the modification,
and black lines represent those after the modification.}\label{fig.flow2}
\end{figure}

These transformations generate a flow of amount $r(t)$ from $t$ to $T\setminus \{t\}$
in the graph capacitated by $\bar{x}^* + x'$.
This indicates that $x' \in \cut(\f_{\bar{x}^*},\allone)$.  
Assigning unbounded capacity to each node in $V\setminus T$,
a similar proof can be derived for $h=\g$.
\end{proof}

\section{Relationship between terminal backup and multiflow}
\label{sec:half_integralty_of_minimum_cost_multiflow}

In this section, we limit the constraints on the generalized terminal backup problem
to the edge-connectivity constraints, unless otherwise stated.
Furthermore, our discussion of multiflows assumes
that edges alone are capacitated.
Let $\Afam$ denote the set of paths connecting distinct terminals,
and assume that the capacity constraints and flow demands are satisfied by 
a multiflow $\psi \colon \Afam \rightarrow \Rset_+$,
i.e., $\sum_{A \in \Afam \colon e \in E(A)} \psi(A) \leq u(e)$ for each $e \in E$
and $\sum_{A \in \Afam_t} \psi(A) \geq r(t)$ for each $t \in T$.
We call a vector (or a function) 
\emph{$1/k$-fractional} if each entry multiplied by $k$ is an integer.

In this section,
we answer the question: to what extent 
the edge-connectivity terminal backup differs from the 
minimum cost multiflow problem in the edge-capacitated setting?
The differences are small, as demonstrated below.

\begin{lemma}\label{lem.fractionality}
For each $1/k$-fractional multiflow, 
there exists a $1/k$-fractional vector of the same cost in $\cut(\g,u)$.
For each $1/2k$-fractional vector $x$, where
$x$ is minimal in $\cut(\g,u)$ and 
$x(\delta(v))$ is $1/k$-fractional for each $v \in V \setminus T$, there exists a $1/2k$-fractional
multiflow $\psi$
such that $x(e)=\sum_{A \in \Afam \colon e \in E(A)}\psi(A)$.
\end{lemma}

The former part of Lemma~\ref{lem.fractionality} is straightforward to prove;
if $\psi$ is a $1/k$-fractional multiflow,
then $x \colon E \rightarrow \Rset_+$ defined by $x(e)=\sum_{A \in \Afam \colon e \in E(A)} \psi(A)$
is $1/k$-fractional and belongs to $\cut(\g,u)$.

To prove the latter part, we use a graph operation called \emph{splitting off}.
Let $e=uv $ and $e'=u'v$ be two edges incident to the same node $v$.
\emph{Splitting off $e$ and $e'$} replaces both $e$ and $e'$ by a new edge $uu'$.
In this section, we regard $\g$ as a set function.
To avoid confusion, we denote $\g$ defined from $r \colon T \rightarrow \Zset_+$ by $\g_r$.
Let $J$ be an edge set on $V$ such that 
\begin{equation}\label{eq:admissible}
|\delta_J(X)|\geq \g_r(X)
\text{ for each } X \in 2^V.
\end{equation}
We say that a pair of edges in $J$ incident to the same node is \emph{admissible} (with respect
to $\g_r$)
when \eqref{eq:admissible} holds
after splitting off the edges.

\begin{lemma}\label{lem.admissible-splitting}
Let $J$ be an edge set on $V$ that satisfies \eqref{eq:admissible},
and let $v$ be a node in $V\setminus T$ with $|\delta_J(v)|\neq 3$.
Then $\delta_J(v)$ includes 
an admissible pair with respect to $\g_r$ or 
\eqref{eq:admissible} holds even after an edge is removed from $\delta_J(v)$.
\end{lemma}

Lemma~\ref{lem.admissible-splitting} derives from a theorem in~\cite{Nutov2009,Bernath2012},
which gave a condition for admissible pairs in a more general setting.
Bern{\'a}th and Kobayashi~\cite{Bernath2014} 
proved an almost identical claim when discussing the degree-specified version
of the edge-connectivity terminal backup, but did not explicitly specify 
the condition under which admissible pairs can exist.
For completeness, we provide a proof of
Lemma~\ref{lem.admissible-splitting} in Appendix~\ref{sec:splitting}.

{\em Proof of Lemma~\ref{lem.fractionality}.}
The former part of Lemma~\ref{lem.fractionality} has been proven above.
Here, we concentrate on the latter part.
Since $x$ is $1/2k$-fractional, $2kx(e) \in \Zset_+$ for each $e \in E$.
Let $J$ be the set of $2kx(e)$ edges parallel to $e$ for each $e\in E$.
Since $x(\delta(X)) \geq \g_r(X)$ for each $X \in 2^V$,
$J$ satisfies 
\begin{equation}\label{eq:cut-J}
|\delta_J(X)| \geq 2k \g_{r}(X)=\g_{2kr}(X) \text{ for each } X \in 2^V.
\end{equation}
Let $v \in V\setminus T$. Since $x(\delta(v))$ is $1/k$-fractional, $|\delta_J(v)|$ is an even integer. 
By the minimality of $x$,
no edge can be removed from $\delta_J(v)$ without violating
\eqref{eq:cut-J}.
Hence, by Lemma~\ref{lem.admissible-splitting}, $\delta_J(v)$ includes an admissible pair
with respect to $\g_{2kr}$.
For each $v \in V \setminus T$,
we repeatedly split off admissible pairs of edges incident to $v$
until no edge is incident to $v$.
The graph at the end of this process is denoted by $(V,J')$.
In $J'$, no edge is incident to nodes in $V\setminus T$, and 
at least $2kr(t)$ edges join $t \in T$ to other terminals.
An edge joining terminals $t$ and $t'$ in $J'$ is generated by
splitting off edges on a path between $t$ and $t'$ in $J$.
In other words, edges in $J'$ correspond to edge-disjoint $T$-paths in $J$.
By pushing a $1/2k$ unit of flow along each of these $T$-paths in $G$,
we obtain the required multiflow.
\qquad\endproof

We see that Theorem~\ref{thm:flow} follows from
Lemma~\ref{lem.fractionality} and the properties of $\cut(\g,u)$ 
described in Section~\ref{sec:characterization}.

{\em Proof of Theorem~\ref{thm:flow}.}
The former part of Lemma~\ref{lem.fractionality} 
implies that $\lp(\g,u)$ relaxes the minimum cost multiflow problem.
As proven in Corollary~\ref{cor.half-integrality},
$\lp(\g,u)$ admits a half-integral optimal solution $x$. 
This solution can be computed in strongly polynomial time
and is guaranteed minimal in $\cut(\g,u)$, as shown in Section~\ref{sec.algorithm}.
By Lemma~\ref{lem.degree}, $x(\delta(v))$ is integer-valued for each $v \in V$.
Hence, the latter part of Lemma~\ref{lem.fractionality} implies that 
there exists a half-integral multiflow $\psi$
such that $x(e)=\sum_{A \in \Afam: e \in E(A)} \psi(A)$.
Note that $\sum_{e \in E}c(e)x(e)=\sum_{A \in \Afam}c(A)\psi(A)$, and therefore $\psi$ minimizes
the cost among all feasible multiflows.

How $\psi$ should be computed from $x$ in strongly polynomial time is unknown.
However, because $\sum_{A \in \Afam_t}\psi(A)=x(\delta(t))$, 
$\nu(t)=\sum_{A \in \Afam_t}\psi(A)$ can be computed for each $t \in T$.
Moreover, $\nu(t)$ is an integer for each $t \in T$.
Therefore, as explained in Section~\ref{subsec.intro-multiflow}, this problem reduces to 
minimizing the cost of maximum multiflow,
for which a strongly polynomial-time algorithm is known~\cite{Karzanov94}.
\qquad\endproof

Each vector $x \in \cut(\f,u)$ belongs to $\cut(\g,u)$. Hence,
we can show that each minimal extreme point of $\cut(\f,u)$
admits a half-integral multiflow of the same cost which is feasible in the edge-capacitated setting.
However we cannot relate extreme points of $\cut(\f,u)$ to feasible multiflows 
in the node-capacitated setting as we observed for star graphs in Section~\ref{subsec.intro-multiflow}.

\section{Conclusion}
\label{sec.conclusion}

We have presented $4/3$-approximation algorithms for the generalized terminal backup problem.
Our result also implies that the integrality gaps of 
the LP relaxations are at most $4/3$.
These gaps are tight even in 
the edge cover problem (i.e., $T=V$ and $r\equiv 1$): 
Consider an instance in which $G$ is a triangle with unit edge costs;
The half-integral solution $x$ with
$x(e)=1/2$ for all $e \in E$ is feasible to the LPs, and its cost is $3/2$;
On the other hand, any integer solution chooses at least two edges from the triangle; 
Since the costs of these integer solutions are at least $2$, the integrality gap is not smaller than $4/3$ in this instance.

An obvious open problem is whether the generalized terminal
backup problem admits polynomial-time exact algorithms or not. 
It seems hard to obtain such an algorithm by rounding
solutions of the LP relaxations because of their integrality gaps.
For the capacitated $b$-edge cover problem,
an LP relaxation of integrality gap one is known~\cite{schrijver-book}. 
For obtaining an LP-based polynomial-time algorithm for the generalized terminal backup problem,
we have to
extend this LP relaxation for the capacitated $b$-edge cover problem.

Another interesting approach is offered by combinatorial 
approximation algorithms
because it is currently a major open problem to find 
a combinatorial constant-factor approximation
algorithm for the survivable network design problem,
for which 
the Jain's iterative rounding algorithm~\cite{Jain01}
achieves 2-approximation.
The survivable network design problem involves
more complicated connectivity constraints than the generalized terminal backup problem.
Hence, study on combinatorial algorithms for the 
latter problem may give useful insights for the 
former problem.
Recently, Hirai~\cite{Hirai14L-extendable} showed that $\lp(\g,u)$ can
be solved by a combinatorial algorithm. Indeed, he also showed that
his algorithm can be used to implement our $4/3$-approximation algorithm for the edge-connectivity
terminal backup without generic LP solvers.

Many problems related to multiflows also remain open. We have shown that 
an LP solution provides
a minimum cost half-integral
multiflow that satisfies the flow demand from each terminal
in the edge-capacitated setting. However, how the computation should proceed in 
the node-capacitated setting remains elusive.
Computing a minimum cost integral multiflow under the same constraints 
is yet another problem worth investigating. We note that
Burlet and Karzanov~\cite{BurletK98} solved a similar problem related to integral multiflows
in the edge-capacitated setting.
Their problem differs from ours in the fact that
$\sum_{A \in \Afam_t}\psi(A)$ is required to 
match the specified value for each terminal $t$.

\section*{Acknowledgements}
This work was partially supported by Japan Society for the Promotion of
Science (JSPS), Grants-in-Aid for Young Scientists (B) 25730008.
The author thanks Hiroshi Hirai for sharing information on multiflows
and his work in~\cite{Hirai14L-extendable}.

\appendix

\section{Proof of Lemma~\ref{lem.admissible-splitting}}
\label{sec:splitting}
Since Lemma~\ref{lem.admissible-splitting} is trivial when $|\delta_J(v)| \leq 2$,
we here suppose that $|\delta_J(v)| \geq 4$.
Assuming that no edge in $\delta_J(v)$ can be removed without violating \eqref{eq:admissible},
we prove that an admissible pair exists in $\delta_J(v)$.

We denote 
$V\setminus \{v\}$ by $V'$,
$\delta_J(v)$ by $A$,
and $J \setminus A$ by $J'$.
For each $X \in 2^{V'}$,
we let $\bar{X}$ denote $V' \setminus X$,
and define $p(X)$ as 
$\max\{\g_r(X),\g_r(\bar{X})\} - |\delta_{J'}(X)|$.
Note that $p$ is a symmetric skew supermodular function on $2^{V'}$.
$J$ satisfies \eqref{eq:admissible} if and only if
$|\delta_A(X)| \geq p(X)$ for each $X \in 2^{V'}$.
The assumption implies that
each $e \in A$ is incident to some $X \in 2^{V'}$
such that $|\delta_A(X)| = p(X)>0$.
A pair of $uv,u'v \in A$ is admissible if and only if no $X \in 2^{V'}$ satisfies
$u,u' \in X$ and $|\delta_A(X)|\leq 1+p(X)$.
We call $X \in 2^{V'}$ a \emph{dangerous set} when
$2 \leq |\delta_A(X)| \leq 1+p(X)$.

If $X$ is a dangerous set,
then $p(X)\geq 1$. Since $p(X) \geq 1$ implies
$\g_r(X)\geq 1$ or $\g_r(\bar{X})\geq 1$,
we have $|X \cap T|=1$ or $|\bar{X} \cap T|=1$ for such $X$.
Without loss of generality, we assume that
each $t \in T$ admits $X \in 2^{V'}$ with $t\in X$ and $p(X)>0$
(otherwise, it suffices to prove the lemma after removing $t$ from $T$).
We denote 
$\{X \in 2^{V'} \colon X \cap T = \{t\}\}$ by $\Cfam'(t)$,
and the set of $X \in \Cfam'(t)$ attaining $\min_{X \in \Cfam'(t)}|\delta_{J'}(X)|$
by $\Mfam(t)$.
Since $\max\{\g_r(X),\g_r(\bar{X})\}=r(t)$ for all $X \in \Cfam'(t)$,
we have $p(Y) \geq 1$ for each $Y \in \Mfam(t)$.
Since $J$ satisfies \eqref{eq:admissible},
$|\delta_A(Y)|\geq 1$ for each $Y \in \Mfam(t)$.

\begin{lemma}\label{lemm.uncrossing-core}
 Let $t,t' \in T$ with $t \neq t'$.
 \begin{itemize}
  \item[\rm (i)] If $X,Y \in \Mfam(t)$, then $X \cap Y,X\cup Y \in \Mfam(t)$.
  \item[\rm (ii)] If $X \in \Mfam(t)$ and $Y \in \Mfam(t')$, then
	$X \setminus Y \in \Mfam(t)$ and $Y \setminus X \in \Mfam(t')$.
  \item[\rm (iii)] If $X$ is minimal in $\Mfam(t)$ and $Y \in \Mfam(t')$,
	then $X \cap Y=\emptyset$.
 \end{itemize}
\end{lemma}
\begin{proof}
 It is known that
 $|\delta_{J'}(X)|+|\delta_{J'}(Y)| \geq |\delta_{J'}(X \cap
 Y)|+|\delta_{J'}(X \cup Y)|$
 and
  $|\delta_{J'}(X)|+|\delta_{J'}(Y)| \geq |\delta_{J'}(X \setminus Y)|+|\delta_{J'}(Y \setminus X)|$
 hold for any $X,Y \in 2^{V'}$.
 If $X,Y \in \Mfam(t)$, then $\g_r(X)=\g_r(Y)=\g_r(X\cap Y)=\g_r(X\cup
 Y)=r(t)$.
 If $X\in \Mfam(t)$ and $Y\in \Mfam(t')$ with $t\neq t'$,
 then $\g_r(X)=\g_r(X\setminus Y)=r(t)$ and $\g_r(Y)=\g_r(Y\setminus X)=r(t')$.
  (i) and (ii) follow from these properties.
 (iii) is indicated by (ii).
\end{proof}

(i) implies that a minimal node set and a maximal node set in $\Mfam(t)$
are unique.
We denote the minimal node set in $\Mfam(t)$ by $Z_t$, and the maximal
node set in $\Mfam(t)$ by $W_t$.

In previous work~\cite{Nutov2009,Bernath2012}, it was shown that $A$ includes an admissible pair
if $p(X) \geq 2$ holds for some $X\in 2^{V'}$.
Hence, in the following discussion, we assume that $p(X)\leq 1$ for each $X \in 2^{V'}$.
By this assumption, 
$p(X)=1$ holds if and only if $X \in \bigcup_{t \in T}\Mfam(t)$. Moreover,
$X$ is a dangerous set if and only if
$|\delta_A(X)|=2$, and 
$X$ or $\bar{X}$ belongs to $\bigcup_{t \in T}\Mfam(t)$.

First, let us prove by contradiction that $|T|\geq 4$.
For this purpose, we suppose that $|T| \leq 3$.
As mentioned above, for each $e \in A$, there exists
$X \in 2^{V'}$ such that $\delta_A(X)=\{e\}$, and $X \in \bigcup_{t \in T}\Mfam(t)$
or $\bar{X} \in \bigcup_{t \in T}\Mfam(t)$ holds.
We let $X_e$ denote one of such $X$.
Because $|A|\geq 4$,
there exist $t \in T$ and distinct edges $e,g \in A$
such that $X_{e} \in \Mfam(t)$ or $\bar{X}_{e} \in \Mfam(t)$, and
$X_{g} \in \Mfam(t)$ or $\bar{X}_{g} \in \Mfam(t)$.
If both $X_{e}$ and $X_{g}$  belong to $\Mfam(t)$,
then $\delta_A(Z_{t}) \subseteq \delta_A(X_{e}) \cap
\delta_A(X_{g})= \emptyset$. Since this contradicts $|\delta_A(Z_t)|
\geq p(Z_t)=1$, 
$\bar{X}_{e} \in \Mfam(t)$ or $\bar{X}_{g} \in \Mfam(t)$ holds.
Without loss of generality, let
$\bar{X}_{e} \in \Mfam(t)$.
Then $X_{e} \cap T = T \setminus \{t\}$.
Since $\bar{X}_{e} \subseteq W_t$,
$Z_{t'} \subseteq X_{e}$ holds for each $t' \in T\setminus \{t\}$.
We notice that $\emptyset \neq \delta_A(Z_{t'}) \subseteq \delta_{A}(X_{e})$
holds for each $t' \in T\setminus \{t\}$,
and $\delta_A(Z_{t'}) \cap \delta_A(Z_{t''})=\emptyset$ holds for each
$t',t'' \in T\setminus \{t\}$ with $t' \neq t''$.
Since these facts imply $|\delta_{A}(X_{e})| \geq 3$,
they contradict the definition of $X_e$.
Therefore $|T|\geq 4$.

Let $t_1,t_2 \in T$ with $t_1\neq t_2$, $e_1 \in \delta_A(Z_{t_1})$, and
$e_2 \in \delta_A(Z_{t_2})$.
Suppose that the pair of $e_1$ and $e_2$ is not admissible.
Then, there exists a dangerous set $Y$ with $\delta_A(Y)=\{e_1,e_2\}$.
$Y \in \Mfam(t_3)$  or $\bar{Y} \in \Mfam(t_3)$ for some $t_3 \in T$.
In the former case,
if $t_3 \neq t_1$,
the existence of $e_1 \in \delta_A(Y)\cap \delta_A(Z_{t_1})$
contradicts $Y \cap Z_{t_1} = \emptyset$,
and 
if $t_3 = t_1$,
the 
existence of $e_2 \in \delta_A(Y)\cap \delta_A(Z_{t_2})$
contradicts $Y \cap Z_{t_2} = \emptyset$.
Hence,
$\bar{Y} \in \Mfam(t_3)$.
Existence of $e_1$ and $e_2$ implies that 
$Z_{t_1} \setminus \bar{Y} \neq \emptyset \neq Z_{t_2} \setminus \bar{Y}$.
If $t_3 \in \{t_1,t_2\}$, the minimality of $Z_{t_1}$ or $Z_{t_2}$ is violated.
Hence, $t_3 \not\in \{t_1,t_2\}$. Now, let $t_4 \in T \setminus \{t_1,t_2,t_3\}$,
and $e_4 \in \delta_A(Z_{t_4})$.
Since
$e_4 \in A \setminus \delta_{A}(Y)=\delta_A(\bar{Y})$, 
we obtain $\bar{Y}\cap Z_{t_4} \neq \emptyset$, which also presents a contradiction.

\end{document}